\documentclass[conference]{IEEEtran}

\usepackage{pdfpages}

\usepackage{afterpage}

\usepackage{amsmath}
\usepackage{amsthm}
\usepackage{msc}
\usepackage{xfrac}

\usepackage{tikz}
\usetikzlibrary{arrows,automata}

\usepackage{graphicx}
\usepackage{subcaption}

\usepackage[section]{algorithm}
\usepackage{algpseudocode}
\usepackage[]{algorithmicx}
\usepackage{xspace}

\usepackage{booktabs}

\usepackage{amsmath, amssymb}
\usepackage{amsthm}
 
 \newcommand{\norm}[1]{\left\lVert#1\right\rVert}
 
\theoremstyle{definition}
\newtheorem{definition}{Definition}[section]

\theoremstyle{theorem}
\newtheorem{theorem}{Theorem}
 
\theoremstyle{remark}

\newcommand{\power}{\textsc{PowerAlert}\xspace}

\begin{document}
\title{\power: An Integrity Checker using Power Measurement}
\author{\IEEEauthorblockN{Ahmed M. Fawaz, Mohammad A. Noureddine, and William H. Sanders}
\IEEEauthorblockA{Department of Electrical and Computer Engineering \\
University of Illinois at Urbana-Champaign\\
\{afawaz2, nouredd2, whs\} @illinois.edu
}
}


\maketitle

\begin{abstract}
We propose \power, an efficient external integrity checker for untrusted hosts.
Current attestation systems suffer from shortcomings in requiring complete checksum of the code segment, being static, use of timing information sourced from the untrusted machine, or use of timing information with high error (network round trip time).
We address those shortcomings by (1) using power measurements from the host to ensure that the checking code is executed and (2) checking a subset of the kernel space over a long period of time.
We compare the power measurement against a learned power model of the execution of the machine and validate that the execution was not tampered.
Finally, power diversifies the integrity checking program to prevent the attacker from adapting.
We implement a prototype of \power using Raspberry pi and evaluate the performance of the integrity checking program generation. We model the interaction between \power and an attacker as a game. We study the effectiveness of the random initiation strategy in deterring the attacker. The study shows that \power forces the attacker to trade-off stealthiness for  the risk of detection, while still maintaining an acceptable probability of detection given the long lifespan of stealthy attacks. 
\end{abstract}




\section{Introduction}
Computer systems are managing most aspects of our lives including critical infrastructure, communication, finance, and health care.
Those computers enable better control of the systems achieving more efficiency while promising reliability and security.
In reality, the promise of security is elusive and is often disrupted by the new exploits and attacks.
Over time, the attacks are getting more sophisticated, targeted, and elusive.
Since the most secure systems are bound to be compromised, intrusion resiliency as a protection strategy has a better chance at improving security.
The resiliency strategy considers compromises inevitable; it moves to detect attacks and devises methods to control a compromise while maintaining an acceptable level of service.
Thus critical to such a strategy is the ability to detect compromises and devise methods to find optimal responses that modify the system to maintain security and service goals.

In this work, we tackle the problem of software integrity checking in the face of Advanced Persistent Threat (APT).
APTs are sophisticated, targeted attacks against a computing system containing a high-value asset~\cite{apt:holmes}.
APTs slowly perform a series of steps in order to gain access and perform an attack. The sequence of steps is often called the kill chain~\cite{killchain} and starts from social engineering to gain credentials,  command-and-control through backdoors, lateral movement to find the high-value asset, and data exfiltration or manipulation. Known APTs, such a Stuxnet, are slow and stealthy with operations spanning months to years to achieve the desired goal. APTs require meticulous planning and tremendous resources; typically available to nation state actors.

A resiliency strategy is vulnerable to a well-planned adversary. The adversary will manipulate monitoring information, necessary for intrusion detection, leading to a false state of security. We propose a system that tackles the following question:
\emph{How to validate the integrity of software against a slow and stealthy attacker without any trusted components in the machine?}

We believe that any pure software solution to the integrity problem cannot address the problem as it faces the Liar's paradox.
In this paradox a liar states that they are lying; for example declaring ``This statement is false". In trying to verify the statement, we reach a contradiction. In a computer setting, a software assesses the state of the machine where it dwells. The paradox arises when the machine declares itself compromised. Even though an attacker does not have an incentive to declare a compromised state, the paradox remains an issue.

For a checker to avoid the paradox, it should (1) be independent of the machine to be checked and (2) use a trust base that cannot be exploited by an attacker. Today's golden standard in security uses in-machine tamper resistant chips (such as TPM or AMT) that hold secrets to generate chains of trust. Those solutions are dependent on the untrusted machine and are closed source,  making them vulnerable to undiscovered exploits.

In order to address the problem, we propose \power, a low-cost out-of-box integrity checker that uses the physics of the machine as a trust base. Specifically, \power directly measures the current drawn by the processor and uses current models to validate the behavior of the untrusted machine. The intuition is that an attacker attempting evasive maneuvers or deception will have to use extra energy thus drawing extra current. \power uses direct measurement of the current and avoids using sensors on-board the untrusted machine as those could be tampered. Moreover, timing information extracted from the signal is very accurate (as opposed to network round trip time which is dependent on the network conditions). 

\power tackles the classical problem of the static defender. A static defender is always at a disadvantage with respect to an attacker, as the attacker can learn the protection mechanism, adapt, and evade the defender. We use a dynamic integrity checking program (IC-Program) that is generated each time \power attempts to check the integrity of the machine. Diversity of the IC-Program evens the playing field between the attacker and defender.

Each time \power decides to check the integrity of the machine it initiates the \power-protocol. It randomly generates the IC-program and a nonce, and sends the pair to the machine. Meanwhile, it starts measuring the current drawn by the processor. The untrusted machine is expected to load the program, to run it, and to return the output.
\power validates the observed behavior by comparing the current signal to the learned current model. 
The IC-Program traverses a small set of addresses and hashes them using a randomly generated hash function. The output of the IC-Program is similarly validated. We show that a low-cost low-power device can efficiently generate the IC-Programs and that the space of IC-Programs is almost impossible to exhaust.

We reduce the overhead of continuous integrity checking by performing the checks in small batches over a small segment of the state of the machine. We model the interaction between the attacker and a \power verifier using a continuous time game and simulate it over a period of 10 days. The attacker attempts to evade the integrity checks by disabling its malicious activities at randomly chosen time instants. 
Our results show the \power verifier forces 
the attacker into a trade-off between stealthiness and the 
risk of detection. An attacker that wants to remain stealthy would need to disable its activities for longer periods of time, while an attacker seeking longer activity 
periods incurs high risks of detection. We also show that if the verifier is moving too slowly (i.e., not making too many checks), the attacker can achieve stealthiness and a low risk of detection while taking less frequent actions.

The paper is organized as follows: the Problem description (Section~\ref{sec:problem}), the System Model and Threat Model (Section~\ref{sec:sys}),  the \power-Protocol (Section~\ref{sec:proto}), the Generation of IC-Programs (Section~\ref{sec:hash}), the method for current signal processing and model learning (Section~\ref{sec:powerana}), the attacker-verifier game (Section~\ref{sec:game}), the Evaluation (Section~\ref{sec:eval}), a security analysis (Section~\ref{sec:security}), Related Work (Section~\ref{sec:related}), and Conclusion (Section~\ref{sec:conclusion}).

\section{Problem Description}
\label{sec:problem}
We tackle the problem of low-cost trustworthy dynamic integrity checking of software running on an untrusted machine.
The goal of the integrity checker is to detect unwanted changes in the known uncompromised static state of a system, while being resilient to attacker evasion and deception.
\begin{definition}{System State.}
Let $X_t=(x_t(0),x_t(1),\ldots, x_t(n))$ be the state of a system at time $t$. The set of locations $\mathcal{L}$ defines the memory locations (addresses) in the state such that for $l\in\mathcal{L}:X_t(l)=x_t(l)$.
\end{definition}
Let $X_t^g$ be the known uncompromised state of the system and $X_t^r(t)$ be the current state of the system. A location $l_c$ is compromised iff $X_t^g(l_c)\neq X_t^r(l_c)$. Let $L_c\in\mathcal{L}$ be the set of compromised addresses. A machine is compromised iff $|L_c|>0$.

The integrity checker detects if a system is compromised by comparing it to a known uncompromised state at time $t$.
\begin{definition}{Integrity Checker.}
Let $f:(X\times X \times \mathcal{L}) \rightarrow \{0,1\}$ be the integrity checking function. $f(X^g,X^r, L)$ returns $1$ if the any of the locations in $L$ are compromised, otherwise it returns $0$.
\end{definition}

Let $T_c$ be the instance of time at which a system is compromised and $L_c$ be the set of locations of compromise in said system.

The integrity checking problem is concerned with finding a sequence of times $T=(t_1,t_2,t_3, \ldots)$ to design and run an integrity checker $f$ for a subset of locations $L=(L_1,L_2,L_3,\ldots)$, such that the compromise is eventually detected with minimum overhead. That is $f(X^g(t_i),X^c(t_i),L_i)=1$ such that $t_i>T_c$  and $L_i\in L_c$.
Finally, the integrity checking process should be validated using side information $i(t)$ in order to avoid deception by an adversary. 

Two trade-offs between effectiveness and performance emerge due to the selection of addresses to be checked and the sequences of times to perform the checks. 
First, the more addresses are checked the higher the chance that the integrity checker will detect the compromise (due to higher coverage), but that comes at a higher cost for the machine and the checker.
Second, if the frequency of checks is high then the checker has a higher chance of detection, but that comes at a higher cost for the machine.

Finally, the integrity checker should be resistant to attacker deception. The deception can be done either by supplying modified answers, disabling the security checking altogether, or disabling the alerting capabilities (that is in case the checker does detect that alerting is disabled). We study the effect of the strategy of the integrity checker in detecting an attacker that attempts to hide in anticipation of a check in Section~\ref{sec:game}.


\section{System Description}
\label{sec:sys}
In order to address the problem of dynamic integrity checking of software (mainly the static memory in the kernel) on an untrusted machine, we propose \power,  an out-of-box device that checks the integrity of an untrusted machine. In this section, we describe our approach for the solution explaining the architecture of \power, protection assumptions, and the threat model.
\subsection{Solution Approach}
On a high level, \power is a trusted external low-cost box tied to the untrusted machine.
Figure~\ref{fig:arch} shows the architecture of \power. The box runs a verification protocol, \power-protocol, on the untrusted machine.
Briefly, the protocol sends a randomly generated integrity checking program, called the IC-program, and a nonce to the untrusted machine.
The machine runs the program, which hashes parts of the static memory, and returns a response to \power. \power checks the response and compares it to the known state of the untrusted machine.
In order to validate that only the IC-program is running, \power measures the current drawn by the processor of the machine and compares it to the current model for normal behavior.
The power model is specific to the processor model and thus has to be learned for each machine. During the initialization of \power, the machine is assumed uncompromised. \power instruments the machine by measuring the current drawn by the processor while running operations semantically similar to the \power-protocol. \power learns a power model specific to the machine that is later used for validation.
\begin{figure}
\center
\includegraphics[width=\columnwidth]{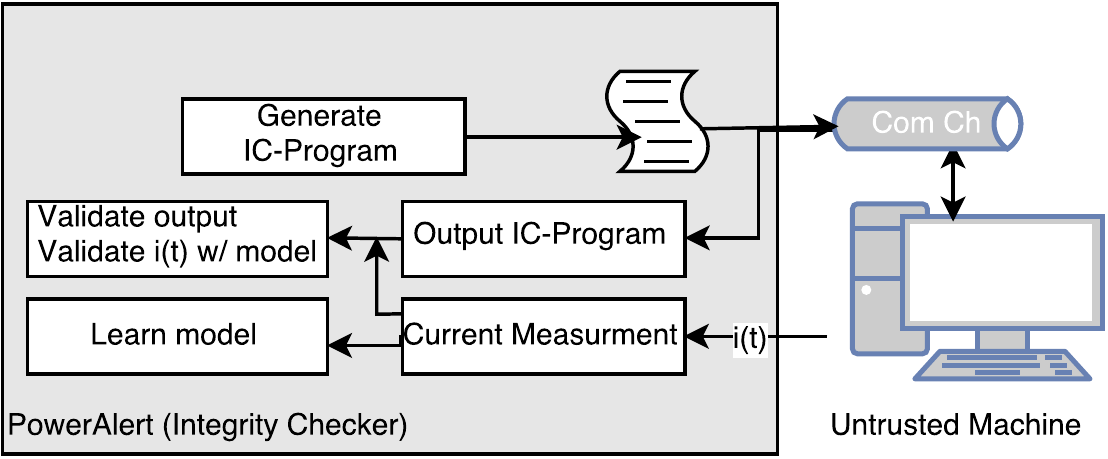}
\caption{The components of PowerAlert.}
\label{fig:arch}
\end{figure}
Even while using the current signal for validation, attackers might evade detection by modifying the integrity checking program.
Researchers typically attempt to implement the most optimized version of the checking program, thus forcing the attacker to incur extra clock cycles for evasion.
Instead, we take a different approach by randomly generating an IC-program every time the \power-protocol is initiated.
The diversity of the IC-program prevents the attacker from adapting, thus making deception harder.
\subsection{Threat Model}
\label{sec:threat}
We assume a fairly powerful attacker when it comes to the untrusted machine; the attacker has complete control over the software. 
However, we assume that the attacker does not modify the hardware of the machine; for example, the attacker does not change the CPU speed, or modify firmware. We do not assume any trusted modules or components on the machine to be tested. Our trust base is derived from the randomness of the protocol and physical properties of the CPU.

We assume that the attacker 
runs deceptive countermeasures to hide her presence and deceive the verifier into a false sense of security, and attempt to reverse-engineer the integrity-checking program for future attempts. Specifically, the attacker can do the following:
\begin{itemize}
\item Reverse engineer the program through static and dynamic analysis. Static analysis allows the attacker to figure out the functionality of the program without running the program. With dynamic analysis, the attacker attempts to understand the functionality of the program if it was obfuscated. 
\item Tamper with the program or the state of the machine in order to supply the checker with the ``correct'' answer. In this case, the attacker needs to understand the functionality of the IC-Program.
\item Impersonate a clean system state by either creating a new program that has similar behavior or using a virtual machine.
\end{itemize}

We aim for our approach to resist the following attacks:
\begin{itemize}
\item \textbf{Proxy attack:} The attacker uses a proxy remote machine with the correct state to compute the correct checksum and returns the result to the verifier.
\item \textbf{Data Pointer redirection:} The attacker attempts to modify the data pointer that is loaded from memory.
\item \textbf{Static analysis:} The attacker analyzes the IC-program to determine its control flow and functionality within the time needed to compute the result. The attacker can precompute and store the results, find location of memory load instructions, or find efficient methods to manipulate the IC-program~\cite{Shaneck2005}.
\item \textbf{Active analysis:} The attacker instruments the IC-Program to find memory load instructions in order to manipulate the program.
\item \textbf{Attacker hiding:} The attacker uses compression~\cite{Li:2011} or ROP storage~\cite{ROP:hide:2009} in data memory to hide the malicious changes when the \power-protocol is running.
\item \textbf{Forced retraining:} The attacker forces \power to retrain models by simulating a hardware fault resulting in a change in hardware.
\end{itemize}

We alleviate those threats by changing the IC-Program on every check, flattening the control structure of the IC-Program, and observing the current trace for abnormal behavior.



\subsection{Assumptions}
In this work, we assume that \power is a trusted external entity. Having \power be an external box as opposed to being an internal module aids in separating the boundaries between the entities. The clear boundary allows us to find a clear attack surface, enables easier alerting capabilities, and easier methods to update \power when vulnerabilities or new features are added. 
Moreover, we assume that the communication channel between \power and the untrusted machine is not compromised. While this assumption can be relaxed by using authentication, we opt to address it in future work.

We assume that \power has a truly random number generator that cannot be predicted by an attacker. Unpredictability is essential for the integrity checking function to be effective. The attacker should not be able to predict the defender's strategy for initiating the \power-protocol. Moreover, the attacker should not be able to predict the IC-Program that will be generated. If the attacker can predict the program, then the attacker can adapt and deceive \power. 

We also assume that \power has complete knowledge of the normal state of the machine. \power uses the known state to verify the output from the untrusted machine. 

Finally, the current measurements are part of our trust based. Those measurements are directly acquired and thus they cannot be tampered with; the learned models are based on the physical properties of the system which cannot be altered. Any attacker computation, such as static analysis of the IC-program, will manifest in the current signal.





\section{\power Protocol}
\label{sec:proto}
We model the interaction between \power and the untrusted machine as an interrogation between a verifier and prover. We name the protocol that defines the interaction the \power-protocol. The goal of the checker is to verify that the prover has the correct proof; in this case, we are interested in the state of the kernel text and data structures. On a high-level, the verifier requests the state of a random subset of the kernel state and the prover has to produce the results. Instead of directly requesting the memory locations, the verifier sends a randomly generated function that hashes a subset of the kernel state. The verifier correlates current measurement and side-channel information with the expected runtime of the sent function.
The \power-protocol is repeated over time; positive results increase confidence that the kernel's integrity is preserved. In the following, we describe the interactions in the \power-protocol.

Figure~\ref{fig:protocol} shows the interactions when the \power-protocol is initiated. At a random instance in time, based on the initiation strategy developed in section~\ref{sec:game}, the verifier initiates the \power-protocol. The verifier starts by randomly generating a hash function $f$, a random function to generate a random set of address positions $\mathcal{L}$, and a nonce $\eta$. In this setting, the hash function $f$ is the IC-program. The verifier connects to the prover and sends the random parameters $< f, \mathcal{L}, \eta>$.
The prover is then supposed to load the hash function, $f$, and run it with inputs $\mathcal{L}$ and $\eta$. Meanwhile, \power measures and records the current drawn by the processor $i(t)$.
Subsequently, the prover sends the output of the hash function back to the verifier. 
Finally, the verifier stops recording the current trace, confirms the output, and validates the expected execution with $i(t)$-- the measured current drawn by the processor.

The verifier introduces uncertainty by changing the hash function, order of-, the subset of-addresses, and the nonce. The uncertainty makes it extremely hard for a deceptive prover to falsify the output. Changing the hash function prevents the attacker from adapting to the verifier's strategy; changing the addresses and nonce prevents the attacker from predicting the verifier's target.

In the following sections, we define the method for generating the hash functions, the strategy for picking a subset of memory addresses, and the method for measuring current and trace correlation.

\section{Integrity Checking Program}
\label{sec:hash}
\power uses the Integrity Checking program, IC-program, to check the integrity of the untrusted machine. We use diversity instead of a static well-designed IC-program.
It is typical for designers to engineer a static well-designed IC-programs. The design typically focuses on hash collision resistance and seek an optimal runtime implementation. A hash collision allows the attacker to find a system state that hashes to the correct value. Thus hash collision resistance ensures that the attacker does not deceive the verifier. 
An optimal runtime implementation ensures that a \emph{memory redirection attack} significantly increases the runtime of the IC-program. A memory redirection attack is when the attacker keeps an uncompromised copy of the state; all memory reads are redirected to the copy.

Once the static IC-program is used the attacker adapts and finds evasive methods to deceive the verifier. A hash function previously thought collision resistant might become vulnerable. Moreover, an implementation once thought optimal might be evaded by an ingenious attacker.
This is the problem of the static defender, the attacker can always find a method to circumvent the protection mechanism. In this arms race, even if the attacker is detected the first time the protection method is revealed, the attacker will adapt and find new methods to hide. An attacker will have enough resources beyond the compromised machine to adapt.

In this work, we take a different approach to address the problem. Instead of building the strongest mechanism possible, we build a changing mechanism that prevents the attacker from adapting.
Specifically, we randomly generate a new IC-program each time the \power-Protocol is initiated and  choose a randomized input set. The input set is drawn randomly from the address space of kernel text and read-only data.

We want to force the untrusted machine to run the IC-program without modification. The IC-program has to be resistant to active and passive (static) analysis:
to counter active analysis, we change the program every time and thus make it hard for the attacker to catch-up;
to counter passive analysis, the program is lightly obfuscated by flattening the control flow structure so that attacker analysis will show up in the power trace. We present the method for 
generating the IC-program in the following sections. 

\subsection{IC-Program Structure}
\label{sec:icprogram}
The IC-Program's purpose is to hash a subset of the state of the untrusted machine, in order to assess the integrity of the machine. The general flow of the program is a loop that reads a new memory location and updates the state of the hash function.

\begin{algorithm}
\caption{IC-Program Pseudocode}
\label{algo:icprogram}
\begin{algorithmic}
\Require Address space size $N$ and nonce $\eta$
\State Initiate hash function with $\eta$, $h=f(h,\eta)$
\For{$n\in [0,N]$}
\State A := generate random address
\State x := load A
\State Update hash, $h=f(h,x)$
\EndFor
\end{algorithmic}
\end{algorithm}

We obfuscate the high-level structure by flattening the control graph of the program using the technique in~\cite{Wang:2001}. The obfuscated program makes it harder for the attacker to locate the load instructions necessary for a memory redirection attack. Any static or active analysis will be observed on the power trace and thus can be detected.

It is important to note that the program that is randomly generated is not polymorphic, that is the functionality of the program changes, not just the structure.


\subsection{Populating the Addresses}
Given the static portion of the kernel's virtual address space which starts from $L_{low}$ to $L_{high}$, a selection algorithm selects an ordered subset of the space for integrity checking.
The ordered subset is a list of address tuples; each tuple contains an address and the number of words to read \verb|<base address, words>| (the size of the word is equivalent to the number of bytes 4/8).
For example, the tuple \verb|<0xffffffff81c000000, 4>| reads 4 words starting from address \verb|0xffffffff81c000000|.
The output of the selection algorithm is a list of the following form: $A=<A_1,k_1>,\ldots <A_j,k_j>$; the expanded form of the list is $A_1, A_1+1, \ldots, A_j+1, \ldots ,A_j+k_j$.
The selection algorithm takes as an input the total number of bytes $N$.
The algorithm generates a random list of address tuples such that $\sum k_i=N$.
The selection algorithm is embedded in the IC-Program; it is simply a Linear Feedback Shift Register (LFSR).
Another important property of the address list is coverage; we define coverage as the fraction of selected bytes over the total size of the system.
Specifically $cov(A)=N/(L_{high}-L_{low})$. The coverage affects the cost of running the IC-Program, and the probability that a given round of \power-Protocol checks a compromised address.

\subsection{LFSR generation}
A new hash function is used for every run of the protocol. We propose chaining randomly generated LFSRs, the outputs of which are combined using a nonlinear Boolean function. Figure~\ref{fig:hash} shows the high-level configuration of the hash function.
Each LFSR is enabled depending on the address being processed. The outputs of the LFSRs are accumulated with the data in a $k$ bit vector. In the following, we explain the method for generating and chaining the LFSRs. The LFSRs are generated using irreducible polynomials in a Galois Field. The configuration of LFSRs is generated using a random tree that specifies the control flow of the program.
\begin{figure}
\includegraphics[width=\columnwidth]{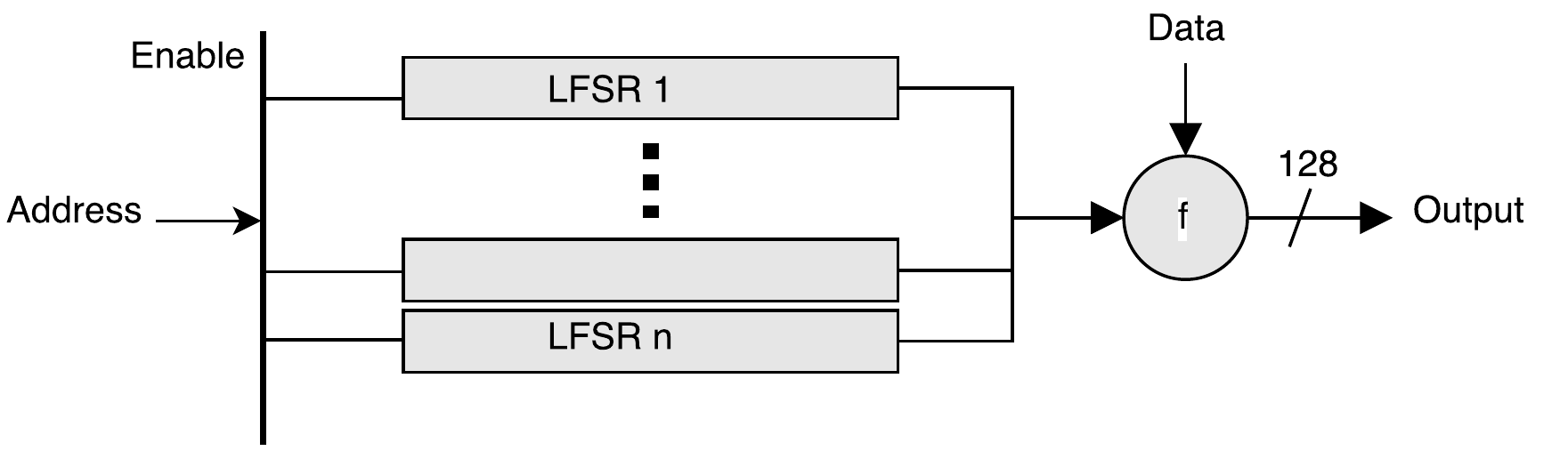}
\caption{General Architecture of the Hash Function}
\label{fig:hash}
\end{figure}

An LFSR is related to polynomials in Galois field $GF(2)$. The process for generating maximal LFSRs uses irreducible polynomial $p(x)$ of degree $n$.
A maximal LFSR has the highest period, the period of the LFSR is the time it takes for the register to return to its initial state. A short period makes it easier to predict the output.
A polynomial is irreducible if $x^{2^n}= x \mod p(x)$.
For a polynomial $p(x)$, the n-bit Galois LFSR is constructed by tapping the positions in the register that are part of $p(x)$. In operation, bits that are tapped get XOR'ed with the output bit and shifted, while untapped bits are shifted without changed. The output bit is the input to the LFSR.
In algorithm~\ref{algo:benor}, we generate random polynomials and apply the Ben-Or irreducibility test~\cite{Gao:1997}. The polynomials are generated by sampling the uniform distribution, $unif(1,2^n-1)$. For example, $p=123$  with binary representation $01111011$ encodes $p(x)=1+x+x^3+x^4+x^5$.
The worst case runtime complexity of the Ben-Or algorithm is $O(n^2\log^2(n)\log\log(n))$. However, the Ben-Or algorithm is efficient as per our experiments in Section~\ref{sec:eval:perf:lfsr}.

\begin{algorithm}
\caption{Irreducible polynomial generation using Ben-Or Irreducibility Test}
\label{algo:benor}
\begin{algorithmic}
\While {true}
\State Generate poly $p(x)\in GF(2)$ of degree at most $n$
\For{i := 1 to $n/2$}
\State g := gcd(p, $x^{2^i}-x \mod p$);
\If {$g \neq 1$}
\State '$p$ is reducible'
\State break
\EndIf
\EndFor
\State\Return '$p$ is irreducible'
\EndWhile
\end{algorithmic}
\end{algorithm}


\subsection{LFSR Chaining}
For a set of $N$ LFSRs, the goal of the chaining strategy is to define the logic for enabling the LFSRs.
The input of the enable logic is the memory address being processed, not the data itself. By using the memory address, the attacker will have a harder time to perform a memory redirection attack. The logic is constructed by creating a random binary tree of depth $n$. The tree defines the control flow of each loop in the IC-program. The control variable at each level is a unique memory address bit.

Each level decides if an LFSR is enabled or not. For each node, an LFSR is enabled/disabled, and then the program counter jumps to either child by comparing an address bit.
\begin{verbatim}
enable LSFR i
if (a[i] == true)  go to child 1
if (a[i] == false) go to child 2
\end{verbatim}
In case the node only has one child, the jump instruction will be omitted for a continuous execution. Figure~\ref{fig:control} shows the structure of the generated tree.
\begin{figure}
\centering
\includegraphics[width=0.8\columnwidth]{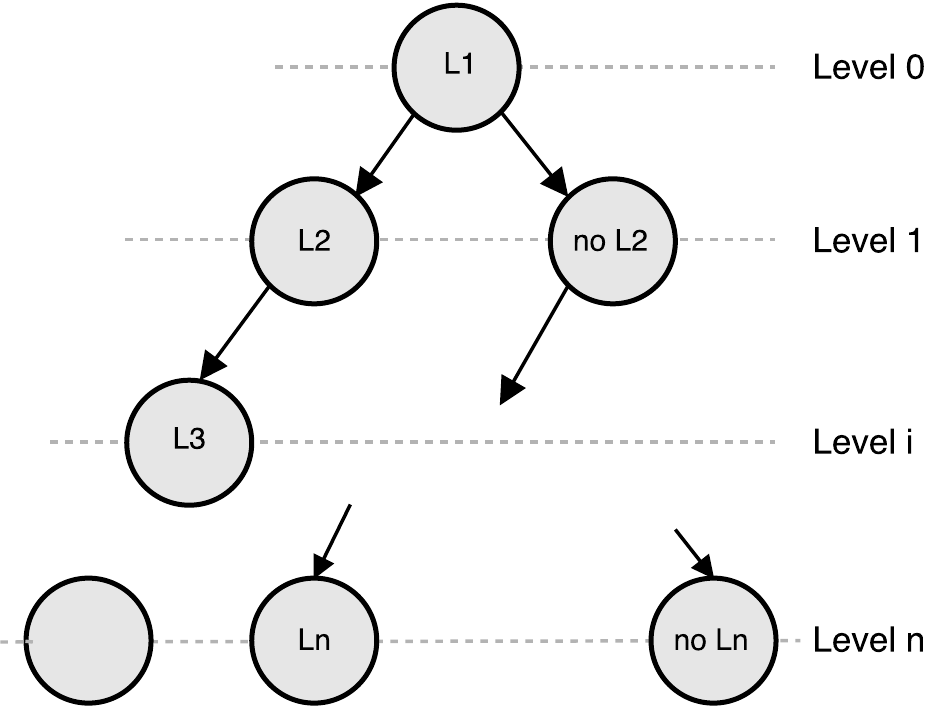}
\caption{Control Structure}
\label{fig:control}
\end{figure}


By providing a program with a randomly generated structure, the attacker cannot predict the program's structure based on previous runs. The attacker is required to perform static analysis on the program (instructions) if she is to attempt to evade the checking algorithm.

\section{Power Analysis}
\label{sec:powerana}
We verify the execution of \power-protocol using the current drawn by the processor. We learn the normal Power Finite State machine (PFSM) model using training data from the machine. Then for each round of the \power-protocol we extract the power states and confirm that they are generated by the normal model.

In the following section we explain the method for current measurement, the method to extract the power states from a current signal, the high-level PFSM model, the method to learn the normal parameters of the model, and the method for validation. Finally, we use the learned model to aid in the parameter selection for IC-program generation.

\subsection{Measurement Method}
The current drawn by the processor is measured using a current measuring loop placed around the line, shown in Figure~\ref{fig:loop}. Our setup works for computers with motherboards that have separate power line for the processor.
Our generation and verification algorithms are not limited to any sampling rate; in fact, the algorithms can be adapted for any sampling rate depending on the needed accuracy.
We measure the current directly by tapping the line from the power supply to the CPU socket on the motherboard; as opposed to measuring the power usage by using the instrumentation provided by the processor as the data will pass through the untrusted software stack.
Such data is susceptible to manipulation and cannot be trusted as an absolute truth.
On the other hand, direct measurement provides a trusted side channel that we use to validate that an untampered \power-protocol is executed.

The measured current signal is either stored for model learning or processed in near real-time for \power-protocol execution validation.
\begin{figure}
\includegraphics[width=\columnwidth]{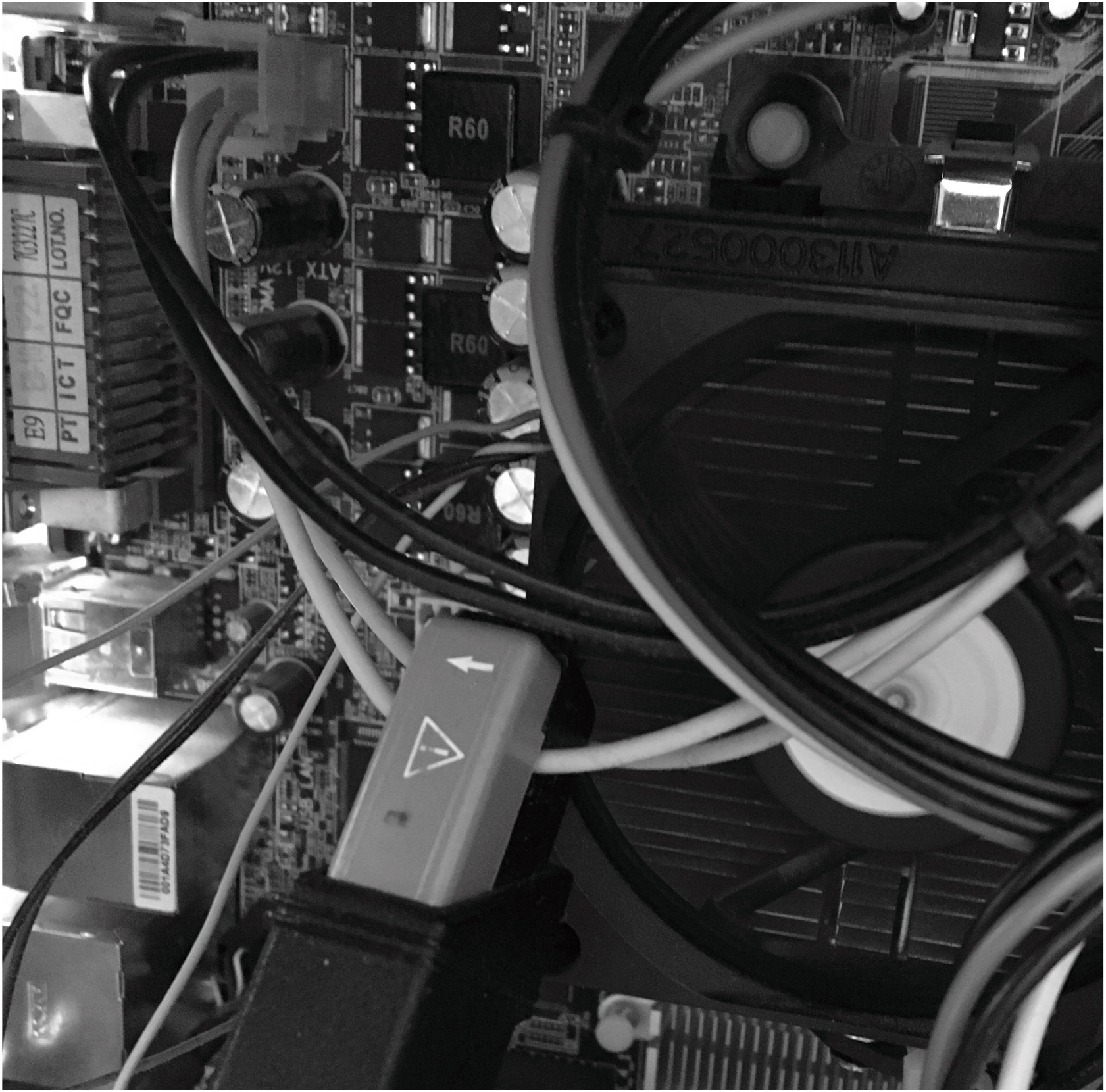}
\caption{The current measurement loop place around the CPU power line.}
\label{fig:loop}
\end{figure}

\subsection{Extracting the Power States}
\label{sec:p:xtract}

We observe that the current drawn by a processor during an operation takes the form of multilevel power states, where each state draws a constant current level. Such behavior is consistent with the way a processor work: different operations use different parts of the processor's circuitry. As each part of the processor switches dynamic current passes through the transistors. Thus different combinations of the circuitry will draw different current levels. 
Thus to learn the power models of operations, we start by extracting the power states exhibited in a current signal.

We start by filtering $i(t)$ using a lowpass filter, $h_1(t)$ to remove high frequency noise from the signal, $i_l(t)=i(t)*h_1(t)$. Then we compute the derivative of the filtered signal, $I(t)=i_l(t)'$. The derivative will be near zero for the pieces of $i(t)$ with a constant current level. We filter the derived signal $I(t)$ with another lowpass filter, $h_2(t)$, to remove more high frequency noise, $I_f(t)=I(t)*h_2(t)$. Finally, we compute a threshold of the signal using an indicator function $I_{>\lambda}(t)$. The indicator function is 1 if the absolute value of a signal is greater than $\lambda$. Figure~\ref{fig:block} shows the block diagram of the transformation.
The transformation leads us to finding the segments of the signal with constant current, those segments are the power states. For each segment $t=[t_a,t_b]$, we compute the average $i_1=\frac{1}{t_b-t_a}\int_{t_a}^{t_b}i_l(t)dt$. The average represents the current drawn during the power state. The duration of each state is computed as $\tau=t_b-t_a$.
\begin{figure}
\centering
\includegraphics[width=0.85\columnwidth]{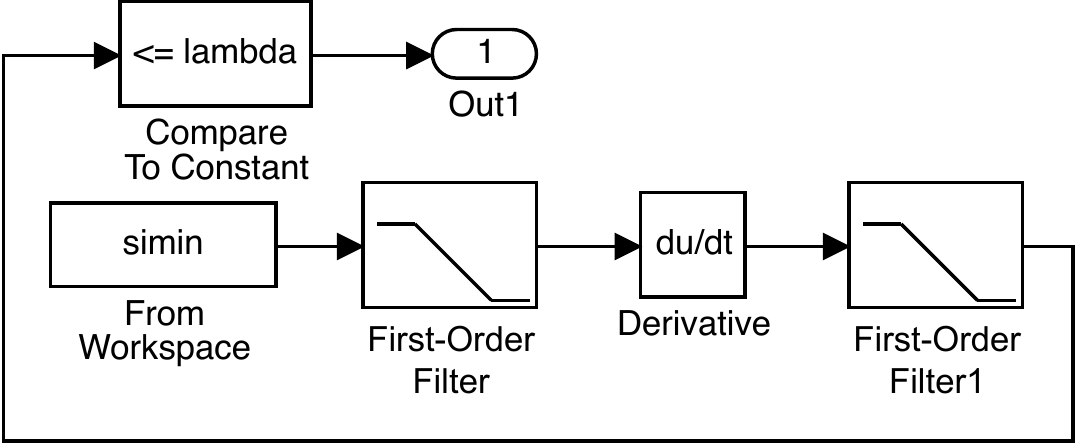}
\caption{Block diagram for power state extraction.}
\label{fig:block}
\end{figure}

\subsection{Power Models}
We model the operations that take place in \power-protocol, network and hashing operations, using Power Finite State Machines (PFSM). The model and timing of each power state are used by \power to validate that the \power-protocol was running untampered.

A Power Finite State Machine (PFSM), proposed by Pathak~\cite{Pathak:2011:FPM}, is a state machine where each state represents a power state $S_k$. Each power state has a constant amount of current drawn $i_k$. The duration of each state is not encoded in the PFSM. A PFSM has an initial idle state $S_0$ with power level $i(S_0)=i_{idle}$. When an operation starts, such as a network receive with a TCP socket, the PFSM moves deterministically to another power state $S_1$ with current level $i(S_1)=i_1$ such that the total current drawn is $i_{idle}+i_1$.

The \power-protocol starts by a network communication (network operation) between \power and the machine. Then the machine is supposed to load and run the IC-program (hash operation). Below is the description of the PFSM of the operations:
\begin{itemize}
\item A short network operation has a PFSM that moves from the idle state $S_0$ to $S_1$ with current level $i_1$. A long network operation alternates between the $S_0$ and $S_1$; the period is $T\ \mu s$. A short network operation is longer than $T\ \mu s$. Figure~\ref{fig:p:net} shows the current trace drawn during a network operation.
\item The hash operation has a a PFSM that moves from the idle state $S_0$ to $S_2$ with current level $i_2$. Then after the hash function is loaded, it starts running and the power state moves to $S_3$ with current level $t_3$. At the end of the operation, the PFSM returns to the idle power state $S_0$. The duration of state $S_3$ is equivalent to the time it takes for the operation to execute. Figure~\ref{fig:p:hash} shows the current trace drawn during a hash operation.
\end{itemize}
We merge the two state machines to follow the operation of \power-protocol. The overall operation state machine is shown in Figure~\ref{fig:psfm:c}.
The \power-protocol PFSM starts from state $S_0$, it moves to $S_1$ during the network operation when the untrusted machine receives the hash function (IC-program), address list, and the nonce. The PFSM then moves to the hash operation, state $S_2$ to load and $S_3$ to run. Finally, as the untrusted machine sends the result to \power the PFSM switches to a network operation, state $S_1$ in particular.

In the following, we explain the method for learning the normal PFSM of a machine.

\begin{figure}
\centering
\begin{subfigure}[b]{0.45\textwidth}
\centering
\caption{Memory read and hash current trace}
\includegraphics[width=\linewidth]{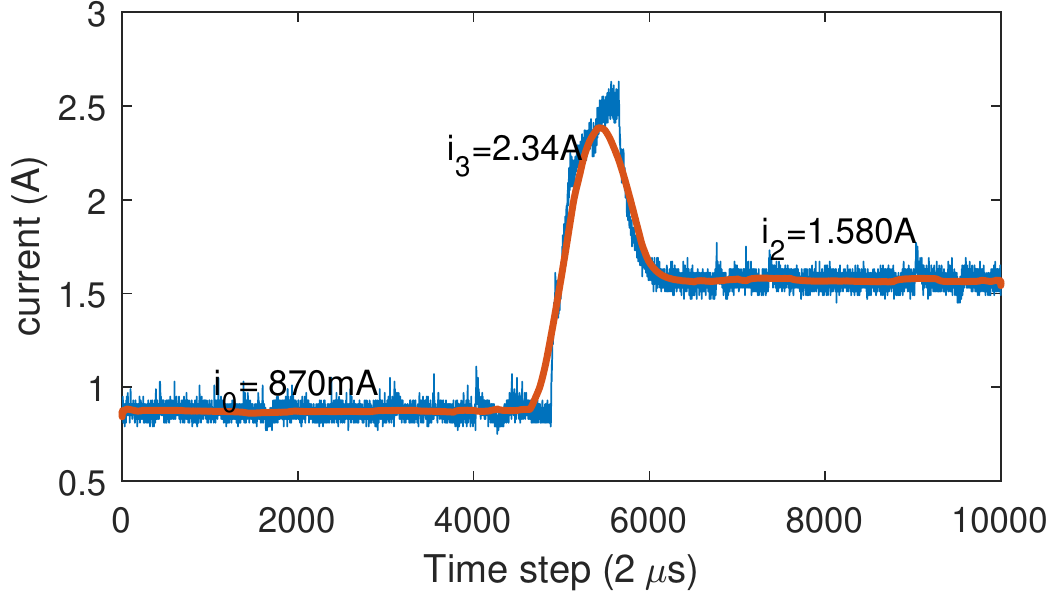}
\label{fig:p:hash}
\end{subfigure}
\begin{subfigure}[b]{0.45\textwidth}
\centering
\caption{Network operation current trace}
\includegraphics[width=\linewidth]{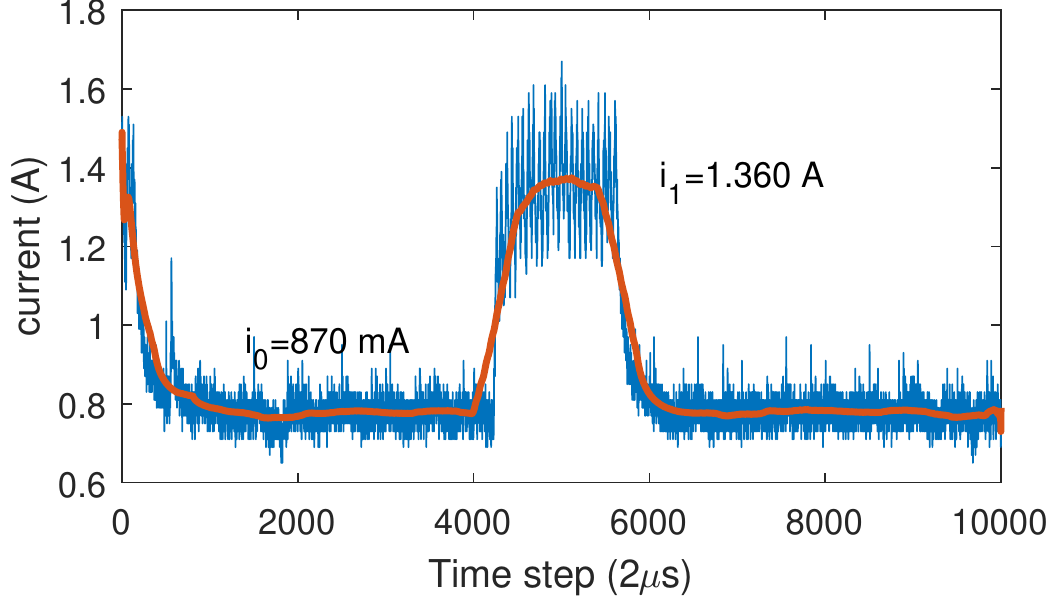}
\label{fig:p:net}
\end{subfigure}
\caption{Current drawn during network and memory read operations.}
\end{figure}
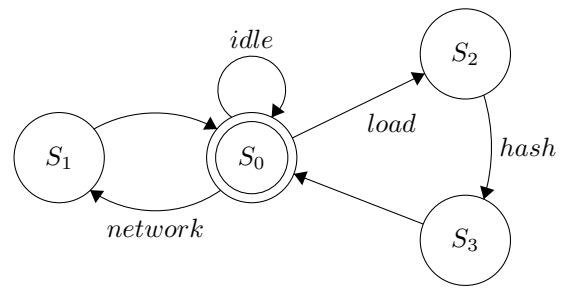
\begin{figure}
\begin{center}
\begin{tikzpicture}[scale=0.2]
\tikzstyle{every node}+=[inner sep=0pt]
\draw [black] (39.6,-24.6) circle (3);
\draw (39.6,-24.6) node {$S_0$};
\draw [black] (39.6,-24.6) circle (2.4);
\draw [black] (53.8,-17.7) circle (3);
\draw (53.8,-17.7) node {$S_2$};
\draw [black] (53.8,-30.1) circle (3);
\draw (53.8,-30.1) node {$S_3$};
\draw [black] (26.8,-24.6) circle (3);
\draw (26.8,-24.6) node {$S_1$};
\draw [black] (42.3,-23.29) -- (51.1,-19.01);
\fill [black] (51.1,-19.01) -- (50.16,-18.91) -- (50.6,-19.81);
\draw (48.9,-21.66) node [below] {$load$};
\draw [black] (55.022,-20.431) arc (16.90507:-16.90507:11.929);
\fill [black] (55.02,-27.37) -- (55.73,-26.75) -- (54.78,-26.46);
\draw (56.04,-23.9) node [right] {$hash$};
\draw [black] (51,-29.02) -- (42.4,-25.68);
\fill [black] (42.4,-25.68) -- (42.96,-26.44) -- (43.32,-25.51);
\draw [black] (38.277,-21.92) arc (234:-54:2.25);
\draw (39.6,-17.35) node [above] {$idle$};
\fill [black] (40.92,-21.92) -- (41.8,-21.57) -- (40.99,-20.98);
\draw [black] (37.547,-26.761) arc (-54.90108:-125.09892:7.56);
\fill [black] (28.85,-26.76) -- (29.22,-27.63) -- (29.79,-26.81);
\draw (33.2,-28.64) node [below] {$network$};
\draw [black] (29.1,-22.699) arc (119.30454:60.69546:8.376);
\fill [black] (37.3,-22.7) -- (36.85,-21.87) -- (36.36,-22.74);
\end{tikzpicture}
\end{center}
\caption{Power Finite State Machine (PFSM) of \power-protocol}
\label{fig:psfm:c}
\end{figure}
\subsection{Learning the models}
For each machine, we assume that we start the from an initial uncompromised state. We establish a power behavioral baseline, build the PFSM and a language for each operation, and learn an execution time model. We initiate the \power-protocol multiple times and store the current signal for every run. For each signal the power states are extracted using the method in section~\ref{sec:p:xtract} and the values of the current drawn are averaged. Moreover, we establish the idle power state by measuring power when no applications are running, that is $i(t)$ is constant.

In our test machine, an AMD Athlon 64 machine running Linux 4.1.13, the current drawn during the idle state is $870mA$, the current drawn during the load phase is $2.34A$, the current drawn during the hash phase is $1.580A$, and the current drawn during the network operation phase is $1.360A$. The idle state current depends on many factors including the services running in the operating system and the semiconductor manufacturing process. The manufacturing process determines static power consumption (subthreshold conduction and tunneling current) which is the current draw when the gates are not switching. Thus the current levels in the generated are unique to the machine and need to be learned for each machine. We decided not fold in semiconductor aging into the power model. Aging causes degradation of the transistor leading to failures; however the time scale where aging affects performance is in the order of 5 years.
Specifically, aging has no effect on dynamic power~\cite{aging:p1} but it does affect threshold voltage. The static power is proportional to the threshold voltage~\cite{zhang2003hotleakage}. Studies have shown that the threshold voltage varies within 1V during thermal accelerated aging~\cite{aging:nsa} which causes a $0.4\%$ increase in static power. We consider this increase insignificant to incorporate into the model especially that it requires years to happen.


In the following, we learn a timing model using the training data from the machine to be inspected and we propose the method of validating the execution of the \power-protocol using the learned PFSM and the timing model.
\subsubsection*{Retraining the models}
In order to retrain the model when needed. We opt for the following procedure: (1) backup the data in permanent storage, (2) wipe storage, (3) install a clean OS, (4) collect training data and learn the models, and (5) restore permanent storage. This process, given our assumption of no hardware attacks, ensures that the attacker cannot interfere with the training process, as the persistent storage is removed during the training phase.

\subsection{Learning Power State Timing} \label{sec:learning}
%
We use timing information in our system as part of the validation process. Specifically, we confirm that an adversary is not trying to deceive \power by extracting the timing information (duration) from each power state and compare it to the learned model, the details are in the next section.
By extracting the timing information using the power signal we control the accuracy of the measure as opposed to using network RTT in remote attestation schemes which are affected by the network conditions.
Moreover, we have confidence that the timing was not manipulated as it was extracted from an untampered source.
We learn the execution time model for the hash phase and the network phase.

The hash function (IC-Program) has a variable number of instructions to execute per cycle and a variable input set size. 
We consider the general structure of the IC-program (Section~\ref{sec:icprogram}) $f(\mathcal{L},\eta)$, where $N=|\mathcal{L}|$ is the size of the input set, and $\norm{f}_c$ is the number of instructions executed by hash function each iteration. All IC-programs have a a complexity $\mathcal{O}(c\cdot N)$ where $N$ is the input size and $c$ is the number of instruction per loop, and use the same type of instructions as any IC-program. We postulate that any IC-program of equal input size $N$ and $c$ number of instructions will have the same execution time. Thus to obtain the training data for learning the timing model, we generate IC-programs for different input size and instruction count and find the execution duration per program.

The experiments are repeated multiple times; the results are averaged to account for the indeterministic nature of program execution. 
We use multivariate linear regression to learn a model of the execution time of the IC-program. The model uses predictor variables $x=[N,\ \norm{f}_c]$ and a response variable $y=t$ (execution time).
For our test machine, Figure~\ref{fig:time:data} shows the data points for duration at power state $S_2$ during the hash phase of the \power-protocol. The surface drawn is the learned model, with implicit equation $y=1.3958 + 0.081 x(1) - 0.017 x(2) + 0.008 x(1)\times x(2)$ and mean error $\sigma=5.4542 \mu s$.
The mean error of the model is significant because it determines the leeway the adversary. If the error is high, then the attacker has a wide gap to employ evasion techniques. However, a minor error means that the attacker has a small gap for evasion. Figure~\ref{fig:rootkit:time} shows the impact of an attacker injecting instructions into the program. The plot shows the current signal measured during the hash phase. The blue signal is the normal behavior, and the orange signal is the tampered one. Both signals have the same power states. However the tampered signal stays longer in the second state.
\begin{figure}
\includegraphics[width=\columnwidth]{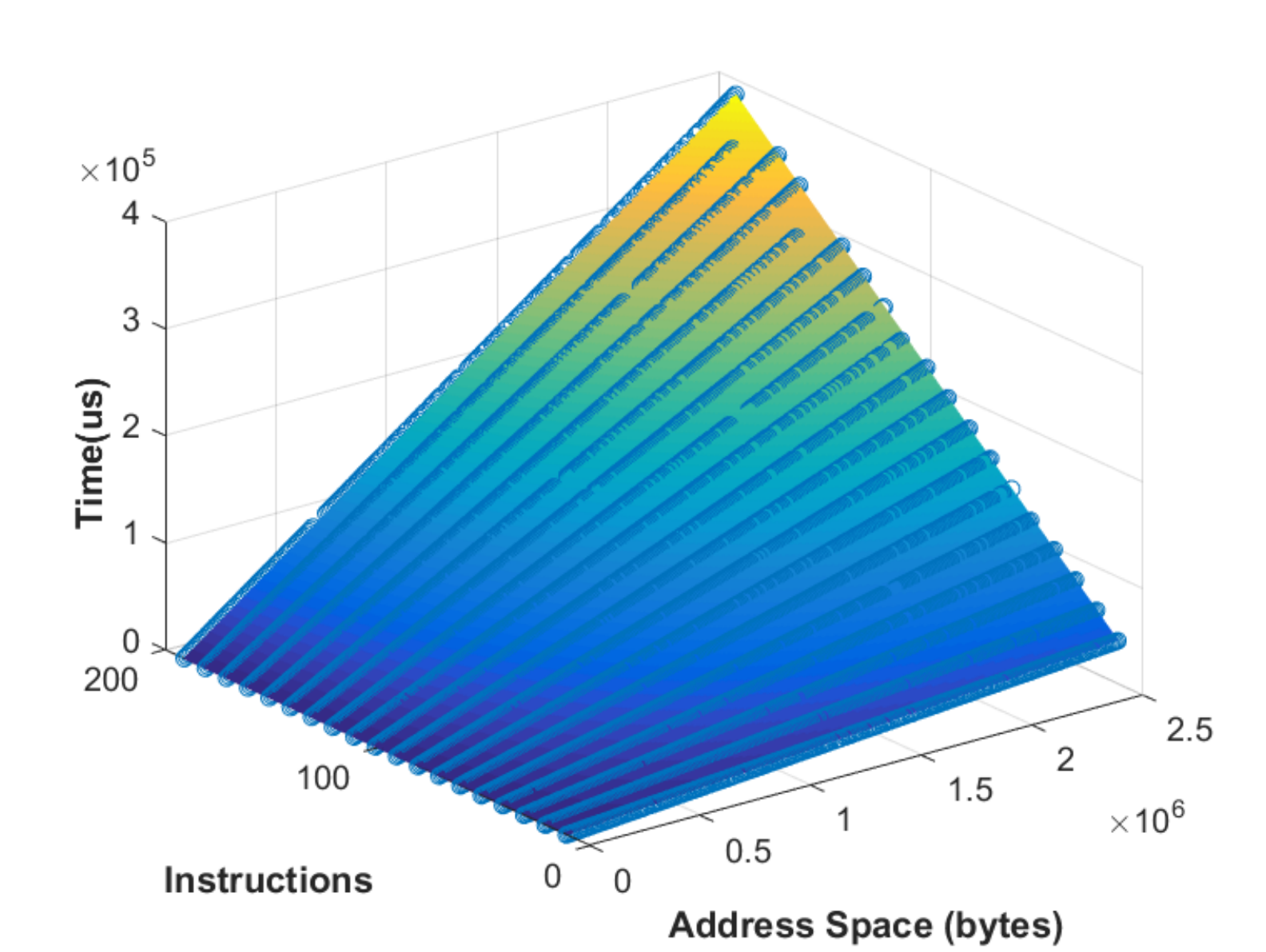}
\caption{Power State Timing model for hashing phase.}
\label{fig:time:data}
\end{figure}
\begin{figure}
\centering
\includegraphics[width=\columnwidth]{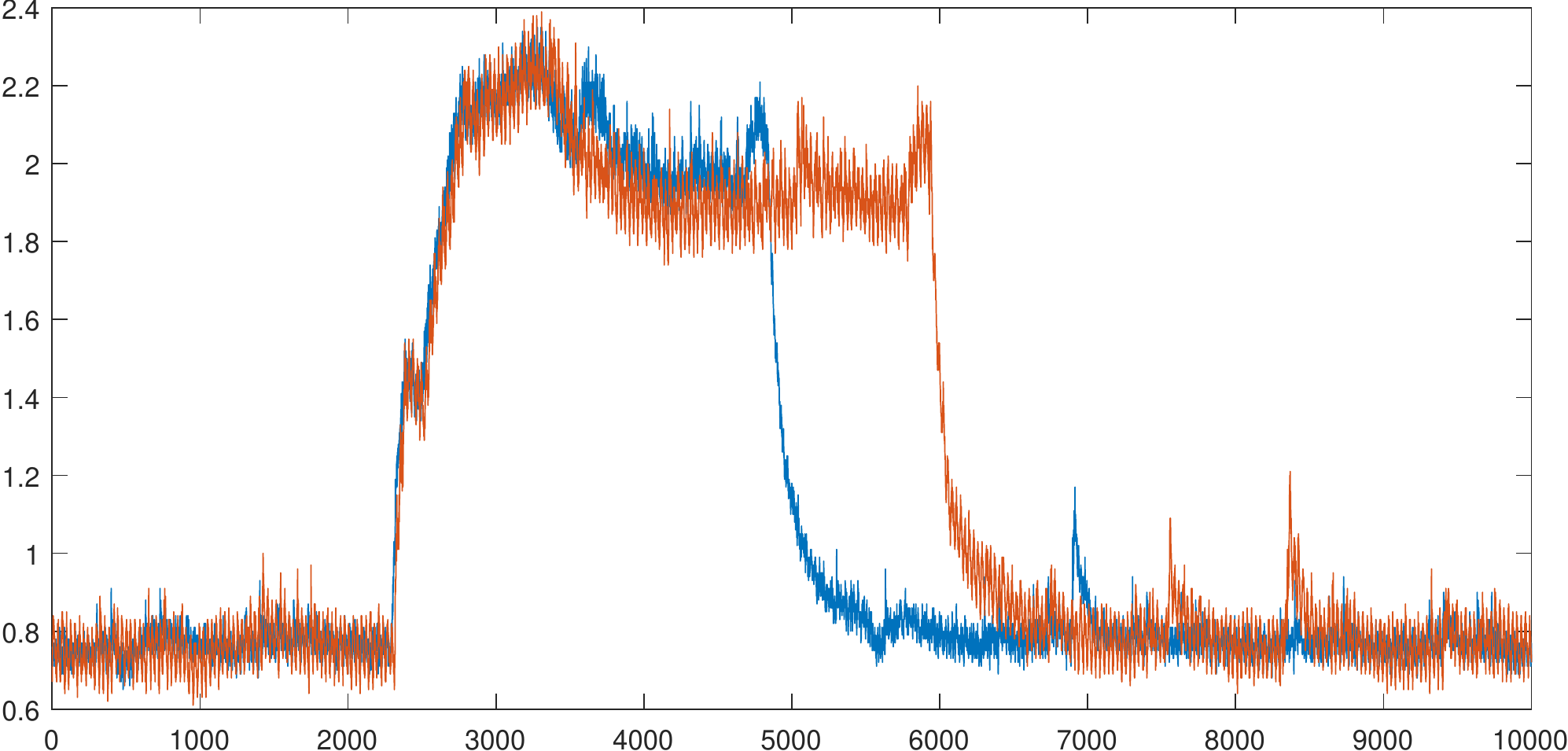}
\caption{Timing difference in current signal due to tampering.}
\label{fig:rootkit:time}
\end{figure}
%
%

During the network phase, the machine is either receiving data or sending the result. When the \power-protocol is initiated, the CPU performs an IO operation to transfer data from the network card. While when the machine returns the results, the CPU performs an IO operation to transfer data to the network card. We learn the timing model of the network phase by varying the number of bytes to be transferred and then measuring the time it takes to transfer those bytes. Using our test machine, $y_n=0.129\times x + 12.48$ is the linear model that predicts the timing for the network phase as a function of the number of bytes ($x$). The mean error of the model is $\sigma_n=1.902\mu s$. The model's constant component is the time it takes the OS to create the buffers and the linear component is the time to transfer the data over the buses using DMA. The linear coefficient models the bus speed.

\subsection{Measurement Validation}
\power measures the current drawn in real-time when it initiates the protocol and attempts to validate that the power trace is generated by the \power-protocol PFSM. The power states are extracted from the current signal resulting in a sequence of states $\mathcal{S}={S(0),S(1),S(2),\ldots, S(n)}$. 
The sequence of states is matched to the regular language which is generated by the \power-protocol PFSM learned by \power during the training phase (Figure~\ref{fig:psfm:c}).
\begin{displaymath}
L=(S_0,S_1)^+ (S_0, S_2, S_3, S_0) (S_0,S_1). 
\end{displaymath}
The first part of the language, $(S_0,S_1)^+$, is the protocol initiation phase. The second part, $(S_0, S_2, S_3, S_0)$, is the hashing phase. Finally, the last part, $(S_0,S_1)$, is the output phase. Note it is expected for the output phase to be a short network operation while the initiation phase to be a long network operation.

During the current trace validation, \power extracts the duration of each power state, and then the duration of each operation phase. We use the timing information for confirming that the duration of each phase is consistent with the protocol operation and hashing.
The expected duration of each phase is determined by the following:
\begin{itemize}
\item For the hash phase: \power computes $\delta = |\hat{y}-y_t|$, the difference between the expected execution time, $\hat{y}$, and the measured execution time, $y_t$. Malicious behavior is detected if $\delta \geq \gamma\max(\sigma_m + \sigma_s)$, where $\sigma_m$ is model error, $\sigma_s$ is the sampling error, and $\gamma$ is a tolerance factor.
\item For the network operation: \power computes $\delta_n=|\hat{y}_n-y_{n,t}|$, where $\hat{y}_n$ is the predicted time and $y_{n,t}$ is the measured time. Malicious behavior is detected if $\delta_n\geq \gamma \max(\sigma_n,\sigma_s)$, where $\sigma_n$ is the model error, and $\sigma_s$ is the sampling error, and $\gamma$ is a tolerance factor.
\item A sanity check is performed to ensure that the total duration of the phases is less than the RTT as measured by \power's clock.
\end{itemize}

The dynamic nature of the hash functions makes the use of timing information effective.
While an attacker can find a faster implementation of a static integrity checking program to evade time-based checks, we, on the other hand, use diversity to prevent the attacker from adapting and thus keeping the timing information a useful method for detecting manipulation.

\subsection{Parameter Search}
\label{sec:optimize}
The gradient of the model, $\nabla y=((0.008x(2)+0.081),(0.008x(1)-0.017))$, reveals that while keeping $x(1)$ constant a slight increase in the number of instructions executed, $x(2)$, leads to an increase in the execution time proportional to the input size, $x(1)$. A larger input size has a greater impact. If an attacker were to inject some instructions into the IC-program, the input size and the original number of instructions determines $\nabla y(x)$, the increase in execution time.

The parameters of the IC-program impact the effectiveness of the detection method. A small $\nabla y(x)$ will increase the false positive rates. While, a large $\nabla y(x)$ will increase the cost to initiate the \power-protocol.
Moreover, the sampling rate of the current measurement system in \power constraints the minimum $\nabla y(x)$ allowed.
The sampling rate is determined by the hardware used, a hardware with low sampling rate has a lower cost than that with a high sampling rate.
A high sampling rate requires more memory and bandwidth to acquire, process, and store the data. 
In order to find the optimal parameters for the IC-program, \power minimizes the total running time constrained by the hardware sampling rate, the cost to run the IC-program, and the coverage required. Table~\ref{tab:min} contains solutions for the optimization below using parameters and models from our test environment. 
\begin{equation*}
\label{eq:optimization}
\begin{aligned}
& \underset{N, \norm{f}_p}{\text{minimize}}
& & y(N,\norm{f}_p) \\
& \text{subject to}
& & y(N,\norm{f}_p+k)-y(N,\norm{f}_p) > \gamma\cdot\max(\sigma_m , \sigma_s) \\
&&& y_n(\norm{f}_p) > \gamma\cdot\max(\sigma_n , \sigma_s) \\
&&& \norm{f}_p < cost,\ N/N_{total} > coverage.
\end{aligned}
\end{equation*}

We compute the parameters of the IC-program for $k=4$ with a tolerance factor $\gamma=10$, $\text{cost}=300$, and $\text{coverage}=0.000001$. We varied the sampling rates for measuring current; the sampling rates reflect the investment made into \power's capabilities. Our computations show that the higher the sampling rate the smaller the IC-program has to be. So if the designer invests more in the hardware capabilities of \power, the attacker's leeway will be tighter and with low overhead to the machine.

\begin{table}
\caption{Minimum IC-Program parameters}
\label{tab:min}
\centering
\begin{tabular}{p{1.5cm} p{1.5cm} p{1.5cm} p{1.5cm}}
\toprule
Sampling Rate & Error Tolerance ($\mu s$) & Coverage (bytes) & Program Size \\
\midrule
 1 MHz      & 64.542 &   2,019      &  40 \\
 500 KHz    & 74.542 &   2,331    &  40 \\
 250 KHz    & 94.542 &   2,956     &  40 \\
 200 KHz    & 104.542 &  3,269    &  40 \\
 54 KHz     & 239.727 &  7,493    &  40 \\
 \bottomrule
\end{tabular}
\end{table}

\section{Attacker-Verifier Game}
\label{sec:game}


In this section, we study the interactions between the \power
strategy and an attacker trying to persist in a target machine.
On a high-level, we introduce a continuous time game 
to model the interactions between the attacker who is trying 
to hide and a verifier using the \power strategy to detect 
intruders. 
%
In this game, the verifier initiates the \power-protocol at random times with a pre-defined strategy. The attacker tries to anticipate the verifier's strategy and disables the malicious changes to the kernel in order to avoid detection.


The verifier's actions consist of deciding on the time instants at which she wants to 
initiate the \power-protocol, while the attacker's actions consist of choosing time instants
at which she wants to hide her activity in order to avoid detection. The attacker's goal
is to coincide her actions with the verifier's actions so that her malicious activity is 
hidden when the \power-protocol is taking place. 
It is in the verifier's interest, however, that the attacker would disable her malicious activities as much as possible.
The verifier's goal, on the other hand, is 
to select a strategy that will catch the attacker off-guard (i.e., when the malicious activity
is not hidden) and detect the attacker's presence. 
In what follows, we first prove that when the verifier chooses her action times independently and identically distributed, then the attacker's best strategy is to 
hide her activities periodically with a fixed period $T^*$. 
We then use simulation
to evaluate the interactions between the attacker and the verifier
for the scenario where the verifier plays according to exponentially distributed attestation times. 
We measure the probability of the attacker being 
detected, the fraction of verifier actions that coincide with the attacker's action, as well 
the fraction of time in which the attacker's malicious activity is hidden, as a function of 
the rates of play of both the verifier and the attacker (the rate $\lambda_0$ of the exponential distribution and the rate $\lambda_1 = \frac{1}{T^*}$ of the periodic distribution). We first 
formalize our game setting, present our theorem, and then present our model and simulation results. 

\subsection{Formalization as a Game}
We follow an approach similar to that presented in FlipIt~\cite{flipit}, a continuous time game
in which both players make stealthy moves (i.e., a player cannot obtain the state of the game unless she makes a move, and thus cannot observe her opponents' moves unless she makes a move of her own) in order to take control of a shared resource. Our formalization differs from FlipIt in that moves
are not instantaneous, they rather spread over an interval of time. Furthermore, the verifier's moves
in our game do not always yield a successful outcome; they often fail either due to the attacker
hiding her malicious activity or due to memory coverage considerations. In what follows we define 
the players' actions, their views of the game, their strategies, 
and the type of strategies we consider in our simulation. We refer by player $0$ to the verifier and by player $1$ to the attacker, and let $\mathcal{B} = \{ \top, \bot \}$ be the set of Boolean constants \texttt{true} 
and \texttt{false}.

\begin{definition}{Attacker action.} The attacker's action consists of hiding her malicious 
activity for a specified period of time. Such an action would restore all of the kernel address 
space location to their original state, thus avoid detection in case the verifier initiates
a \power-porotocol attestation process. Formally, 
let $c: \mathcal{R}^{+} \longrightarrow \mathcal{B}$ be the function defining the state of the attacker's malicious activity at any time $t > 0$, i.e.,
\[
c(t) = \left\{ 
 \begin{aligned}
  \top & \quad \mbox{iff attacker is active at time t} \\
  \bot & \quad \mbox{otherwise}
 \end{aligned}
\right.
\]
Let $\mathcal{C}$ be the state of all such state functions. 
We slightly abuse the notation to write $c([t_a, t_b])$ to refer to the state of the attacker's 
activity in the time interval $t_a \leq t \leq t_b$. 

An attacker's action is therefore defined as a function $a_1: \mathcal{C} \longrightarrow 
\mathcal{C}$ that changes the state of the attacker's activity for a period of time $\alpha_1$. 
Formally,
\[
a_1(c([t, t + \alpha_1]) = \left \{ 
 \begin{aligned}[lr]
  \bot & \quad \mbox{iff } c([t,t+\alpha_1]) = \top \\
  c([t, t + \alpha_1]) & \quad \mbox{otherwise}
 \end{aligned}
\right.
\]

\end{definition}

\begin{definition}{Verifier action.} The verifier's action consists of initiating 
the \power-protocol and attempting to attest the victim machine's kernel address space. 
An attestation fails when the attacker has modified a memory location that the verifier 
is attempting to attest. A successful attestation does not necessarily mean the absence of
malicious activities; the attacker might have changed memory locations not requested for 
attestation by the \power-protocol. 

Formally, the verifier's action is a function $a_0: \mathcal{R}^{+} \longrightarrow
 \mathcal{B}$ where $a_0(t)$ is the outcome of initiating an attestation process at time 
 $t$. $a_0(t) = \top$ if the attestation succeeds and $a_0 = \bot$ if the attestation fails and 
 the verifier detects the presence of the attacker. 
 Let $\alpha_0$ be the length needed to complete an attestation procedure,
and let $p_e$ to the probability that an attacker evades the verifier's 
attestation attempt. In other words, $p_e$ refers to the probability that the verifier's 
action succeeds even in the presence of an active attacker. Therefore for
 $t_v > 0$, we can write 
 \[
 a_0(t_v) = \left \{ 
 \begin{aligned}
  \top & \; \mbox{ if } \exists{t \in [t_v, t_v + \alpha_0)}. C(t) = \top & \mbox{w.p. } p_e \\
  \bot & \; \mbox{ if } \exists{t \in [t_v, t_v + \alpha_0)}. C(t) = \top & \mbox{w.p. } 1 - p_e \\
  \top & & \mbox{otherwise}
 \end{aligned}
 \right.
 \]
\end{definition}
The game ends if the attacker is detected, i.e, if $\exists{t > 0}$ such that $a_0(t) = \bot$. 

\begin{definition}{Feedback function.} Let $t_{i,k}$ be the time at which player $i$ makes her
$k$'th move. We refer by $\phi_i(t_{i,k})$ to the feedback that player $i$ receives when 
she makes an action at time $t_{i,k}$. The verifier can only observe the outcome of her own actions
while the attacker can observe the outcomes of  any action (i.e., attestation attempt) that the verifier has attempted between the time the attacker last move and the current move time. 
Therefore we can write
\[
\phi_0(t_{0,k}) = \{ a_0(t_{0,k}) \}
\]
and
\[
\begin{aligned}
&\phi_1(t_{1,k}) = \{ a_1(t_{1,k}) \} \; \cup \\ 
&\quad \quad \{ a_0(t_{0,j}) \; | \; t_{1,k-1} + \alpha_1 < t_{0,j} < t_{1,k} + \alpha_1 \}
\end{aligned}
\]
 
\end{definition}

\begin{definition}{Player View.} For each player $i \in \{ 0, 1 \}$, we define the player's view of the game at time $t$ as 
$
v_i(t) = \{ \left( t_{i,1}, \phi_i(t_{i,1}) \right), \ldots,
\left( t_{i,k}, \phi_i(t_{i,k}) \right) \}
$,
where $t_{i,k} \leq t$ is the time at which player $i$ made her last move before time $t$.
We denote by $\mathcal{V}$ to be the set of all possible player views. 
\end{definition}

\begin{definition}{Player Strategy.} A player's strategy defines the time instants 
at which she wants to make her move. Formally, let $v_i(t_k)$ be player $i$'s view 
at time $t_{i,k} = t_k$ when she made her $k$'th move, then the player's strategy 
is a function $S_i: \mathcal{V} \longrightarrow \mathcal{R}$ such that
$t_{i, k+1} = t_{i,k} + S(v_i(t_k))$.
\end{definition}

\begin{definition}{Renewal Strategy.} A strategy $S_i$ is a renewal strategy if the 
action inter-arrival times are independent and identically distributed. In other words, the 
action inter-arrival times form a renewal process~\cite{stoch}. 
\end{definition}


\begin{theorem} \label{th:game}
If the verifier is playing with a renewal strategy, 
then the attacker's best strategy is to play periodically with a fixed period $T^*$. 
\end{theorem}

\begin{proof}
Since the verifier's action inter-arrival times are i.i.d., if at any time $t_k$ the attacker 
computes an optimal action play time $S_1(v_1(t_k))$, then $S_1(v_1(t_k))$ would also be 
optimal any other time instant $t_k'$. Therefore the attacker's best strategy to play 
periodically with period $T^* = S_1(v_1(t_k))$. Obtaining the analytical expression of $T^*$ is 
beyond the scope of this paper and considered for future work. 
\end{proof}

Theorem~\ref{th:game} states that the attacker's best response strategy to an attacker playing with a renewal strategy is to play periodically. The theorem further illustrates an important advantage that the attacker enjoys over the verifier, that of observability. Since the verifier does not know if the attacker is present or not, the events of a successful verification attempt and the attacker hiding her activity are indistinguishable. In other words, the verifier cannot observe the attacker's actions, and thus must choose her strategy before playing the game. As for the attacker, the verifier's actions are completely observable, and thus she can use that information to further improve her strategy. 


\subsection{Simulation and Results}
We implement a model of our game using {\em Stochastic Activity Networks} (SAN)~\cite{san}
in the M\"{o}bius modeling and simulation tool~\cite{mobius}. In accordance with theorem~\ref{th:game}, 
we assume that the verifier is playing with an exponential strategy with rate $\lambda_0$, while the attacker plays with periodic strategy with rate $\lambda_1 = \frac{1}{T_1}$, where $T_1$ is the attacker's period.
We vary the players' rates $\lambda_0$ and $\lambda_1$, and
evaluate the performance of 
their strategies with respect to three metrics: (1) the probability of detection, 
(2) the fraction of time the attacker is inactive, and (3) the hit ratio. We compute the {\em probability 
of detection} as the average fraction of simulation runs in which the game has ended. We define the attacker's 
inactivity indicator function as 
\[
I(t) = \left\{ 
\begin{aligned}
1 & \quad \mbox{ iff } C(t) = \bot \\ 
0 & \quad \mbox{ otherwise.} 
\end{aligned}
\right.
\]
For a time period $T$, we can then write the {\em fraction of time the attacker is inactive} as 
$\frac{1}{T} \int_{0}^{T} I(t)dt$.
We finally define the {\em hit ratio} as the fraction of verifier actions that coincide with 
attacker actions. In other words, the hit ratio refers to the fraction of attestation attempts
that succeed because the attacker has turned off her malicious activity.

For our simulations, we assume that the verifier is using a sampling rate of $500$KHz. 
From Table~\ref{tab:min}, we know that the coverage of the verifier's generated programs 
is $2331$ bytes. For a machine running the Linux 4.1.13 kernel with an average kernel memory size of $200$ MB, we can compute the probability of evasion as $p_e = 1 - \frac{2331}{200 * 1024 * 1024} = 0.99998$. Also, using our learned model in Section~\ref{sec:learning}, we can compute the time needed for attestation $\alpha_0 = 903 \mu s$. 

We vary the attacker's period ($T_1$) from $30$ seconds to $5$ minutes in steps of 10 seconds, and the verifier actions'
average inter-arrival times ($T_0$) from 1 minutes to 3 minutes in steps of 15 seconds. Additionally, 
we assume that the attacker chooses to hide her activity for half of her period, i.e. $\alpha_1 = \frac{T_1}{2} = \frac{1}{2\lambda_1}$. We run our simulation for 
10 days and report average results for all of our metrics. 

Figures~\ref{fig:prob_det} and~\ref{fig:time_off} show the probability of detection and 
the fraction of time the attacker is inactive, as a function of the attacker's play rate $\lambda_1$
and the verifier's play rate $\lambda_0$, respectively. 
For a fixed attacker play rate, the probability of detection increases as the verifier is 
increasing her play rate. Intuitively, the verifier is performing more attestation procedures,
and thus the probability that an attestation takes place while the attacker is active increases, so 
it is more likely that the verifier will be able to detect the attacker's presence. 

For a fixed verifier strategy, decreasing the attacker's play rate would yield a reduction
in the probability of detection. As the attacker is playing slower, her period increases, 
and thus she would have to hide her malicious activity for longer periods of time 
(if $\lambda_1 > \lambda_1'$ then $\frac{1}{2\lambda_1} < \frac{1}{2\lambda_1'}$). Therefore it is more
likely that the verifier's attestation attempt will coincide with the time she has her malicious 
activity hidden, thus reducing her probability of detection. This is shown in Figure~\ref{fig:time_off}
where for a fixed verifier play rate, the fraction of time where the attacker turns off her malicious 
activity decreases as her play rate increases. This is further confirmed by Figure~\ref{fig:hit_ratio}
where the hit ratio is at its highest when the attacker is playing slow, and decreases as the attacker
increases her play rate. 
This means that when the attacker is playing slowly, 
she is hiding more often and thus avoiding the verifier's attestations, therefore reducing the probability 
of being detected. This however comes at the expense of her activity time, the slower she plays, the 
longer she has to turn off her malicious activity, reaching a fraction of 70\% of inactivity when the 
verifier is playing at her slowest rate. 

Another interesting result is revealed by the peaks of the surface in Figures~\ref{fig:time_off}
and~\ref{fig:hit_ratio}. The fraction of the time the attacker is inactive reaches its peak 
when the attacker is playing slowly while the verifier is acting at her fastest rate. This is due
to the fact that a high verifier rate induces a high hit ratio, thus during one period of inactivity, 
the attacker might have to go through multiple attestation procedures, and might even have to extend her 
inactivity period if an attestation procedure is started right before the end of that period. 

Furthermore, when the verifier chooses a slow rate of play, the rate of increase in the probability 
of detection as a function of the attacker's play rate is slow. Therefore it is best for the attacker, in this case, to play fast, thus reducing her fraction of inactivity time while maintaining a relatively 
acceptable probability of detection. Therefore a slow rate of play for the verifier puts her at a 
disadvantage where that attacker can play fast, increase her activity time and still avoid detection. 
However, as the defender starts increasing her rate of play, the attacker faces a trade-off between 
the probability of detection and the fraction of inactivity time. Achieving high levels of activity 
means risking higher detection probability while keeping the detection probability low would require
the attacker to remain inactive for longer periods of time. The decision on such a trade-off 
depends on the attacker's level of stealthiness. An attacker who would want to remain stealthy, as in the 
case of an APT, would be more inclined to turn off their malicious activities more often than an 
attacker who's goal is to inflict the maximum amount of damage in the shortest period of time. 

In summary, our simulation results show that the usage of the \power-protocol for checking the integrity
of the kernel space would force a malicious attacker to balance a trade-off between risk of detection 
and periods of activity. An attacker that wishes to remain active as much as possible would risk a higher probability of detection, while an attacker that seeks to remain stealthy would have to 
incur periods of inactivity that can be as high as 70\% of her life time. 

\begin{figure}
 \centering
 \includegraphics[width=\columnwidth]{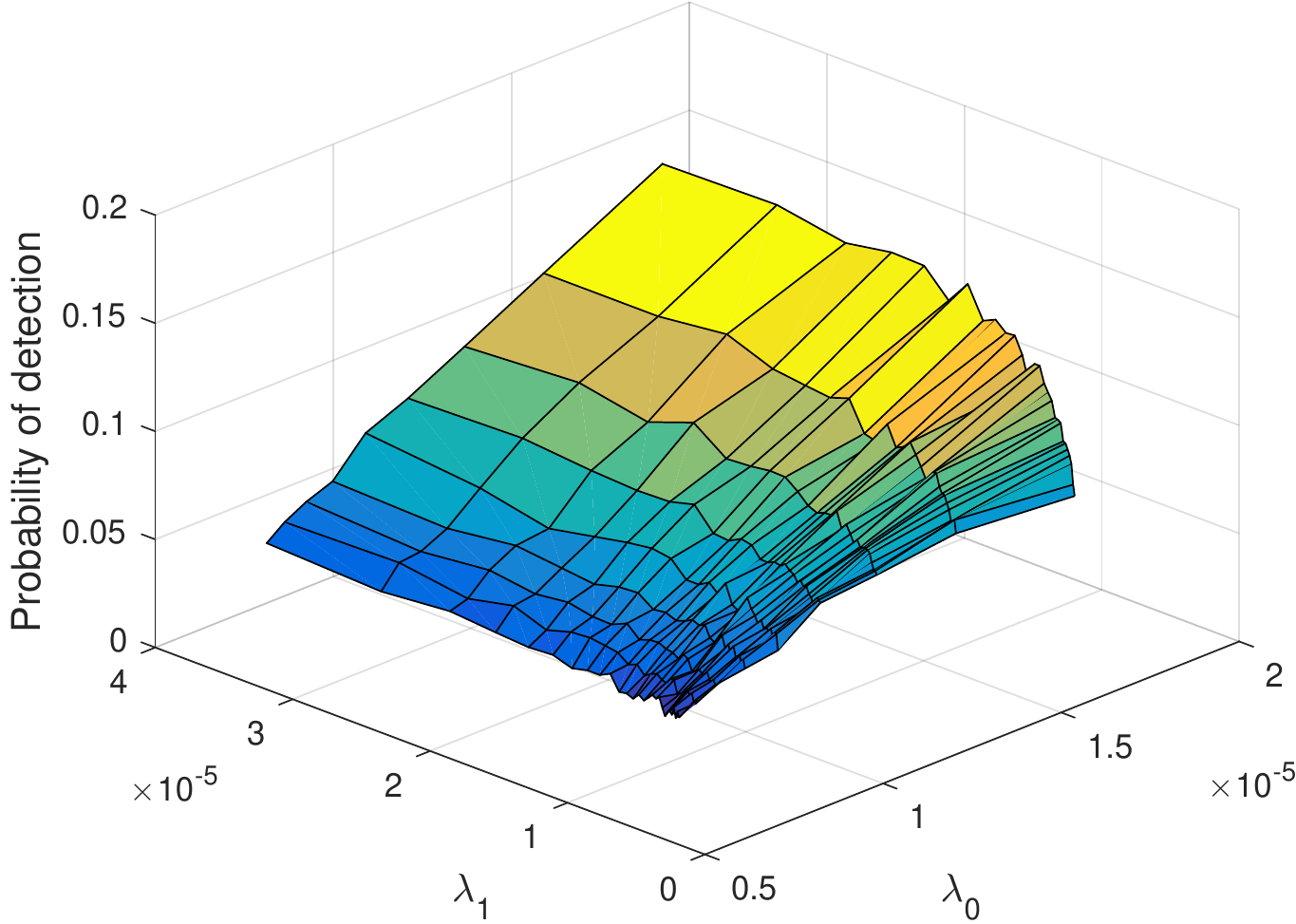}
 \caption{The average probability of the attacker's detection as a function of the attacker's and the verifier's play rates.}
 \label{fig:prob_det}
\end{figure}

\begin{figure}
 \centering
 \includegraphics[width=\columnwidth]{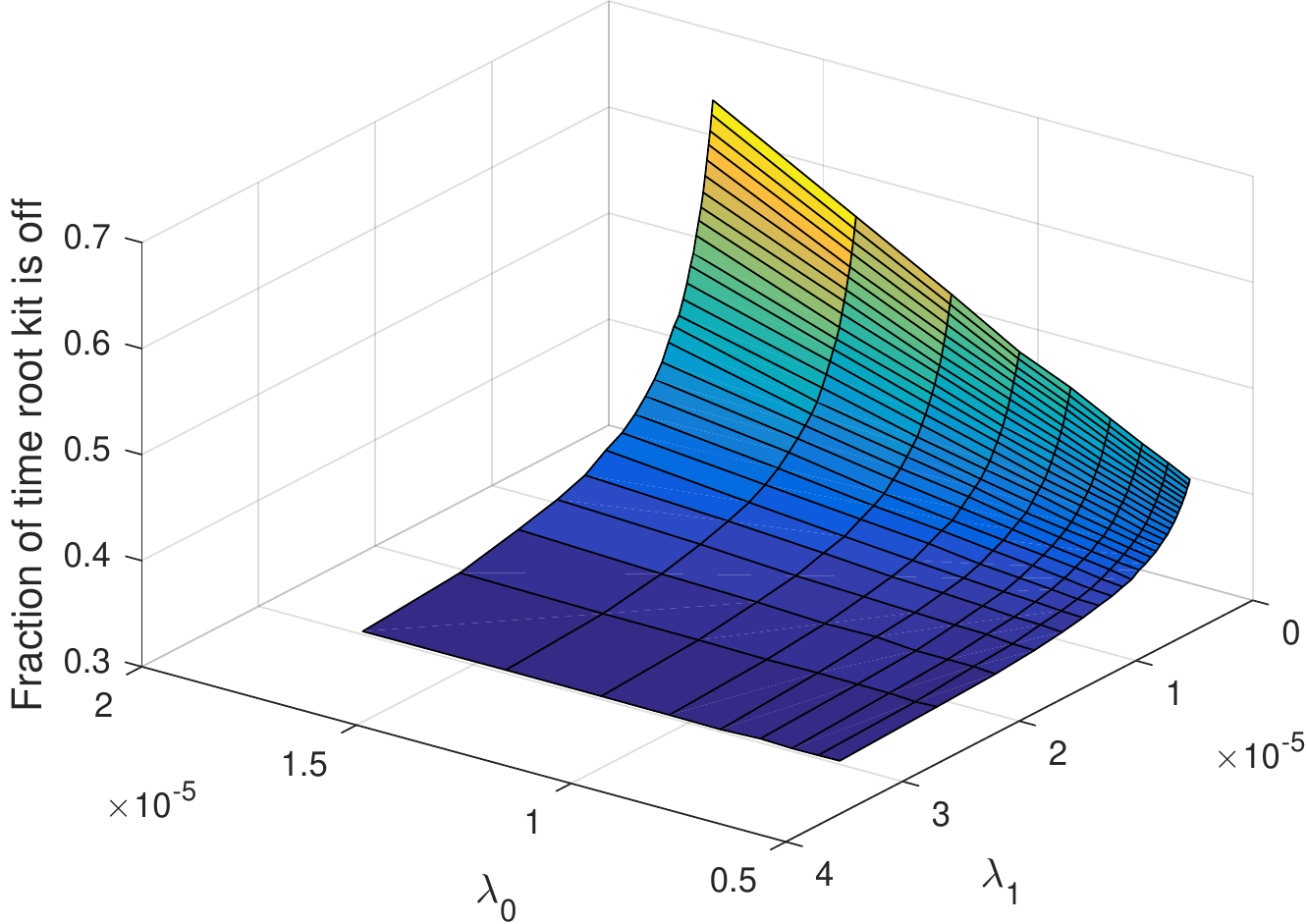}
 \caption{The average fraction of the time the attacker's malicious activity is hidden as a function of the attacker's and the verifier's play rates.}
 \label{fig:time_off}
\end{figure}

\begin{figure}
 \centering
 \includegraphics[width=\columnwidth]{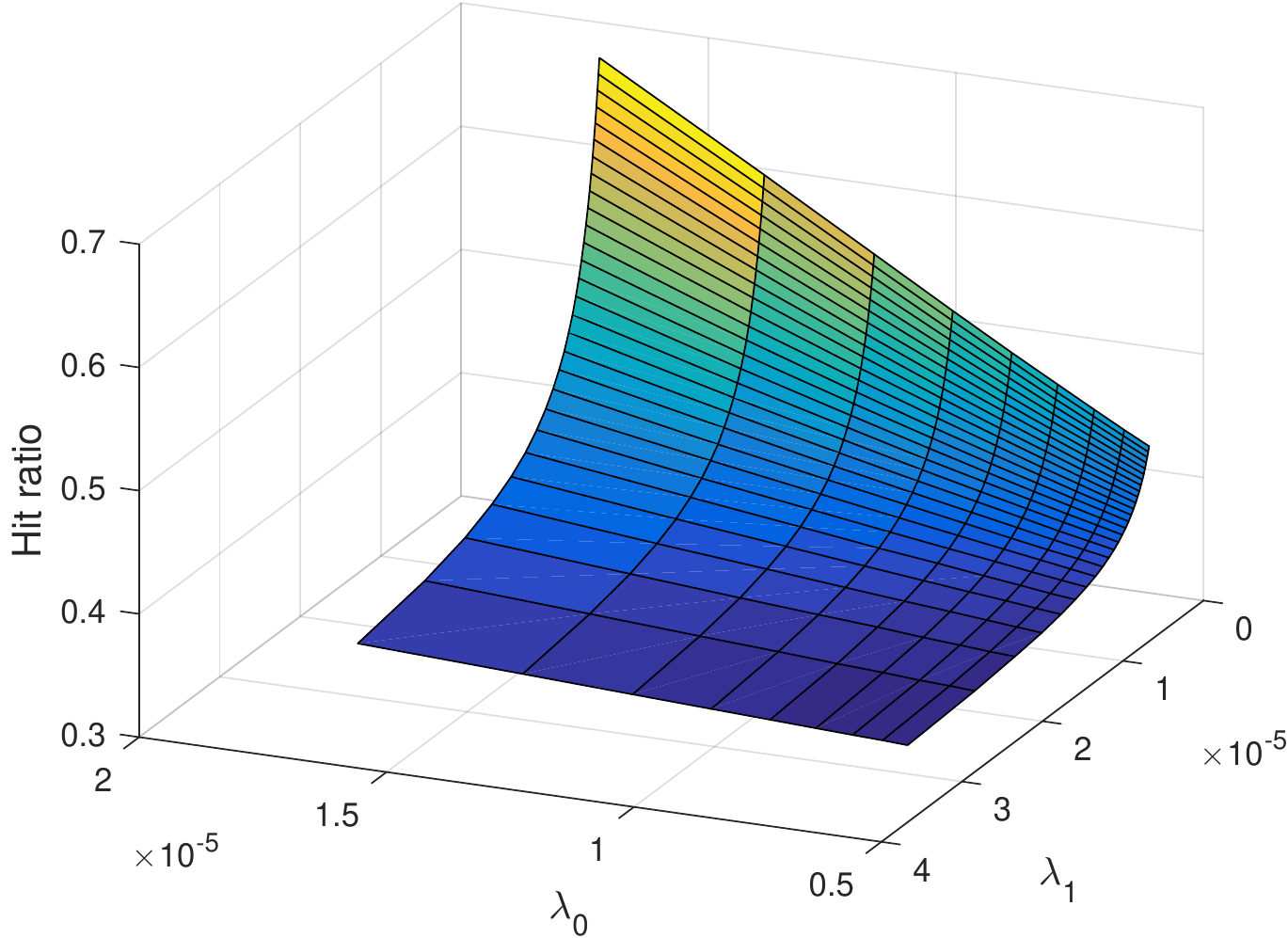}    
 \caption{The average ratio of attestation tasks that were evaded by the attacker's actions as a function of the attacker's and the verifier's play rates.}
 \label{fig:hit_ratio}
\end{figure}

\section{Evaluation}
\label{sec:eval}
In this section, we evaluate the performance of \power in generating the IC-Program and the size of space of IC-Programs.
\subsection{IC-Program Generation}
\label{sec:eval:perf:lfsr}
In the program generation algorithm, the hard problem is the generation of random irreducible polynomials of degree $d$ in $GF(2)$. We implement the generation algorithm using NTL~\cite{ntl}, A Library for doing Number Theory, on a Raspberry PI 2 v1.1. The results in Figure~\ref{fig:lfsrgen} show the average time it takes to generate the LFSRs for each degree $d$ using our implementation compared to the worst case complexity.
Our implementation performs orders of magnitude better than the worst case runtime.
Practically speaking it takes around one second to generate an irreducible polynomial of degree 128.  The performance can be significantly improved by either optimizing the algorithm, parallelizing the generation algorithm, or precomputing and then caching the generated polynomials. The rate of initiation of the \power-protocol should be less than that of the rate of IC-program generation for the system to be stable.
In our simulation we explored the effect of different initiation rates on the defender utility.
\begin{figure}
\includegraphics[width=\columnwidth]{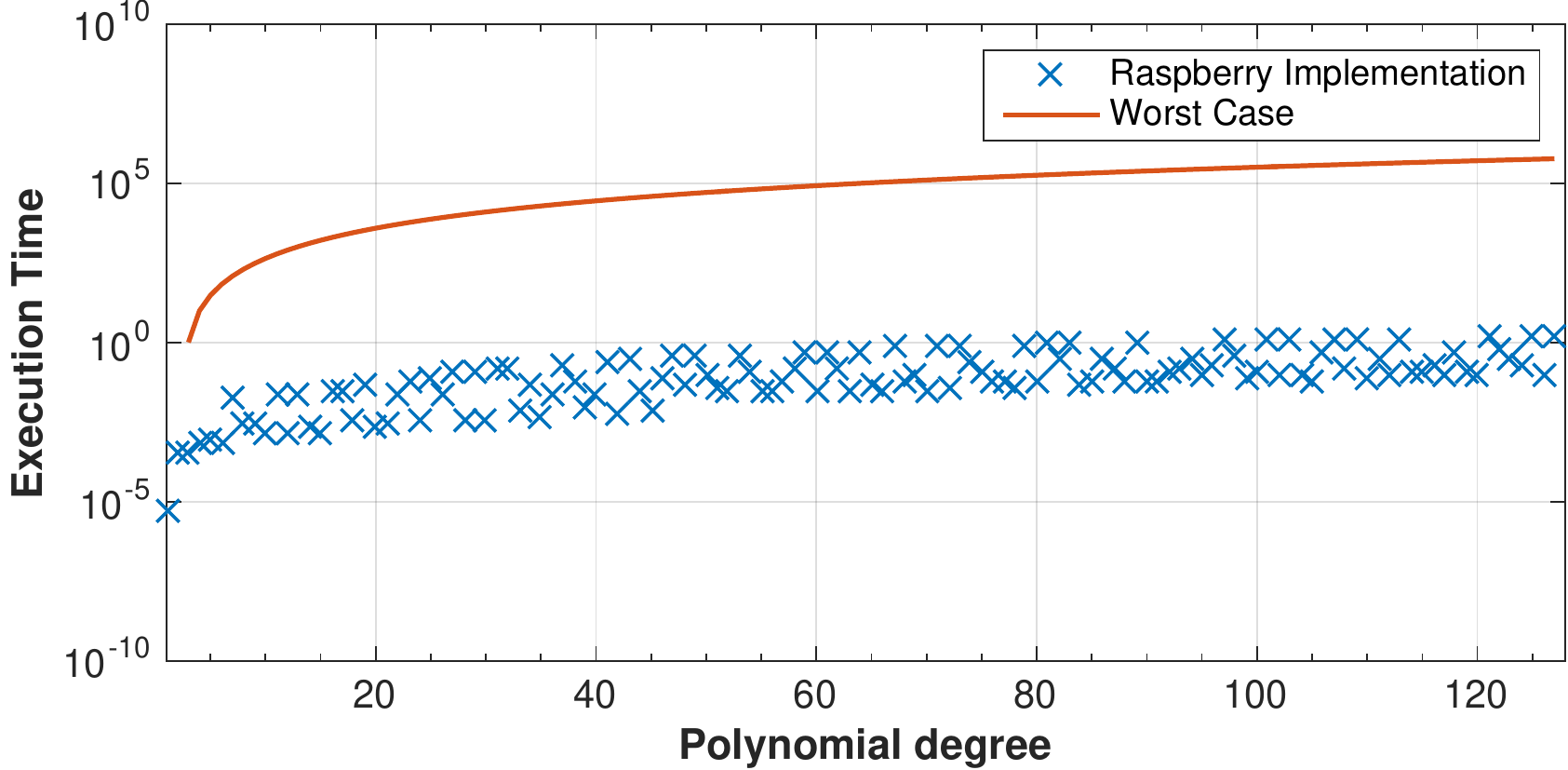}
\caption{Average time in seconds for generating a random irreducible polynomial for degree $d$}
\label{fig:lfsrgen}
\end{figure}
\subsection{Maximum IC-Programs}
We want to investigate the maximum number of IC-programs that can be generated by \power. The goal is to have a large space so that the generated programs are not reused.
An IC-program is generated by chaining randomly generated LFSRs of degree $d$ using a randomly generated binary tree of depth $n$.
The maximum number of IC-programs that can be generated is the product of the maximum number of binary trees multiplied by the maximum number of irreducible polynomials.
Let the maximum number of binary trees with depth $n$ be $t_n$ (equations below). The maximum number of nodes for a binary tree of depth $n$ is $2^n$ and thus the total number of tree is the sum of the Catalan number $C_m$, which the number of binary tree with $m$ nodes, over the total number of possible nodes. Let $M_d$ be the maximum number of irreducible polynomials of degree $d$ in $GF(2)$, $M_d$ is called the necklace polynomial.
\begin{equation}
t_n = \sum_{i=0}^{2^n}\left(\prod_{k=2}^{i}\frac{i+k}{k}\right), M_d(2)=\frac{1}{d}\sum_{k|d}\mu(k)2^{d} \notag
\end{equation}

Finally, the total number of IC-programs that can be generated is $D_{d,n}=M_d(2)\times t_n$. Figure~\ref{fig:total:ic} shows the total number of programs for an increasing degree of polynomials and depth of tree. The maximum number of programs reaches $1.9721\times 10^{26}$ for $n=40$ and $d=5$ guarantees that no program ever gets reused in the lifetime of the device, in fact if a new program is generated every one second the space would be depleted in $6.246\times 10^{18}$ years.
\begin{figure}
\includegraphics[width=\columnwidth]{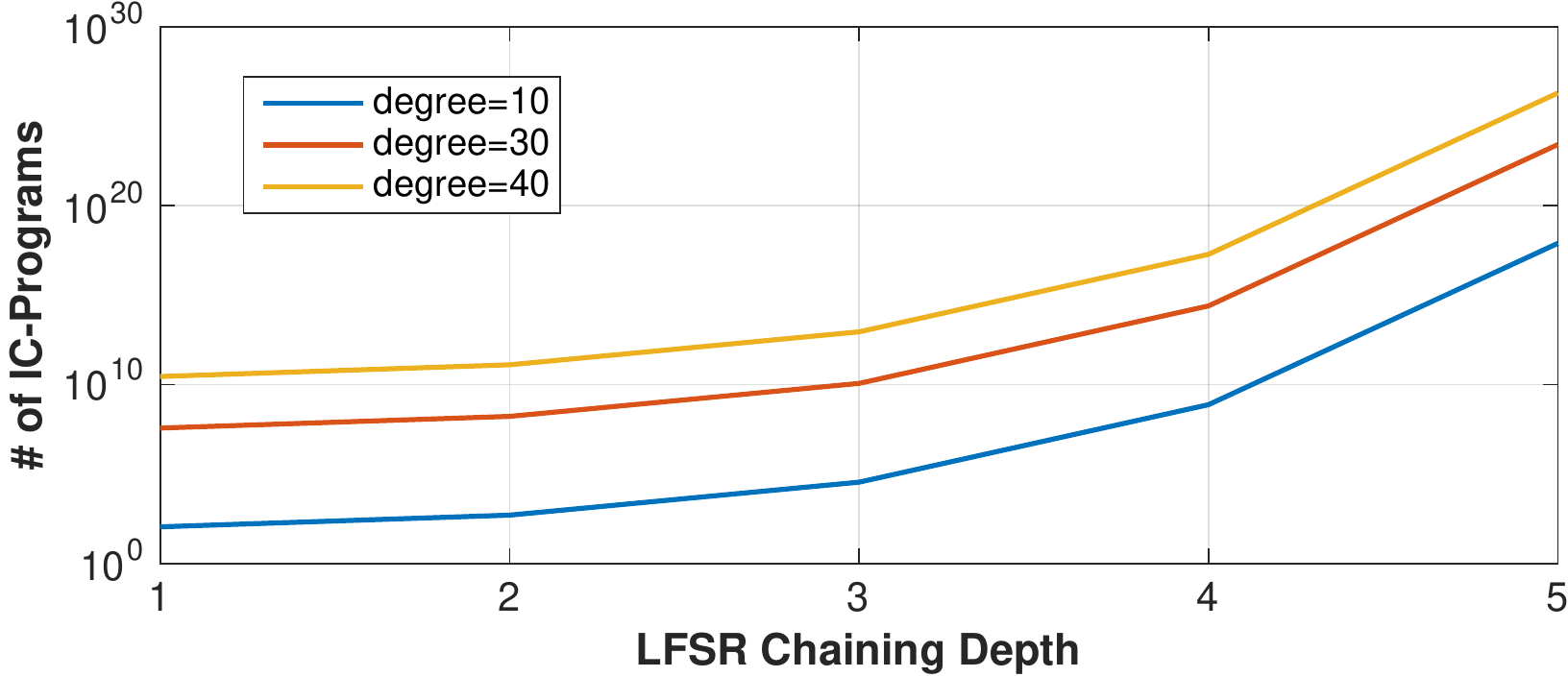}
\caption{Maximum Number of IC-Programs a varying tree depth and polynomial degree.}
\label{fig:total:ic}
\end{figure}

\subsection{Performance Impact}
When the \power-protocol is initiated and the IC-program starts execution, execution of all other tasks is paused. This is needed in order to ensure that no other tasks interfere with the current measured from the CPU.
In terms of graphical responsiveness, the pixel response time should not exceed $4ms$~\cite{Miller:1968:response}. The \power-procotol is initiated, on average, once every minute for $0.9ms$. Thus the graphical degradation will not be noticeable by a user. Moreover, we measured the performance degradation to be a factor of $0.0018$.

\section{Discussion}
In this section, we discuss some security details related to the implementation of this system. Specifically, we discuss the attack surface of \power and the security concerns with the IC-Program. Moreover, we consider the practicality of our solution, and it's important despite the existence of TPMs.

\subsection{Implementation Details}
Each \power device has a client on the untrusted machine. The client is a low-level module that communicates with \power.
The client can be implemented for placement in the kernel or the hypervisor.
The communication channel between \power and the client can be over any medium such as Ethernet, USB, or serial. All those channels are feasible because of the proximity between \power and the untrusted machine.
The use of serial or USB communication is advantageous because it limits the attacker to physical attacks, making man-in-the-middle and collusion attacks harder. If the attacker has physical access to the machine, then she could tamper with \power.

The client receives the IC-Program as machine code over the communication channel. \power signs the code, their keys are exchanged during the initialization phase of the system.
The signed program allows the machine to attest that \power is the generator. We propose using a stream cipher as it has better performance than public-private key ciphers or block ciphers.

  As for the \power's hardware, the resource requirements are minimal. We implemented a prototype using Raspberry PI 2. The prototype uses an ADC to convert the current measurements from the current loop to a digital signal. The ADC uses a sampling rate of 500KHz; most low-cost hardware can handle this sampling rate. Moreover, power state extraction is only performed when the \power-Protocol is initiated, the operation does not need to be real-time.



\subsection{\power's Attack Surface}
In case the untrusted machine gets compromised, an attacker might try to compromise \power to disable its functionality. The attack surface of power is limited to one communication channel that only uses the \power-Protocol. During the protocol \power receives the output of the IC-Program only. The current measurements are out of the control of the attacker.
The language of the protocol is context-free and thus can be verified using Language-theoretic security approaches~\cite{langsec:input:14}.
By verifying the parser, we have an assurance that even if the attacker compromises the machine, it cannot spread to \power.


\subsection{Comparison to TPM}
\power does not rely on specialized hardware within the untrusted machine such as TPM  or Intel's AMT. However, \power and trusted modules are orthogonal systems; whereas TPMs provide a method for secure boot, dynamic integrity checking is still costly and harder to enforce.
\power provides an external security solution that can be tied to a security management across a wide network. In fact, \power can use Intel's AMT as a communication channel.
Finally, our work demonstrates the need for measurements that do not pass through or origin in the untrusted machine. Such measurements reduce the risk of attacker tampering and mimicry.

\section{Security Analysis}
\label{sec:security}
\power uses current measurements, timing information, and diversity of the IC-program to protect against attackers subverting integrity checking. In this section we list how \power addresses the attacks in section~\ref{sec:threat}. In all cases, the power measurements ensure that detection is certain if the attacker was to perform parallel tasks, as she will be drawing more current than expected.



\textbf{Proxy Attack}:
If the attacker attempts to forward the IC-program to a remote machine to compute and return the result via the same network link. \power can detect this attack by examining its effects on the current trace and timing of the network phase. Using the current trace, \power will observe that network operations took longer than expected as more bytes will be transferred between the CPU and the network card. The size of the IC-program, which was picked by the optimization in Section~\ref{sec:optimize}, ensures that our hardware will pickup the re-transmission. For our test machine, Table~\ref{tab:min} shows that 40 instructions per iteration is enough. Any physical attack that such as tapping the network line, or firmware changes to the NIC are not within our purview.

\textbf{Data Pointer redirection attack}:
The attacker stores an unmodified copy of the data in another portion of memory. When an address is to be checked, the attacker changes the address to be checked with the unmodified portion. The IC-program uses the address and the memory content when computing the hash function. In order to compute a valid hash, the attacker has to change the address to the location of the copy while retaining the original address. Conservatively, the attacker has to add at least four instructions ($k=4$) per loop in order to achieve that goal. Table~\ref{tab:min} shows the parameters of the IC-program to ensure detection given our measurement error.

\textbf{Static Analysis}: Analyzing a flattened control flow is NP-Hard~\cite{Wang:2001}. Thus it will not be possible for the attacker to analyze the program without significant computations.

\textbf{Active Analysis}: Active reverse engineering is used to learn the usage patterns of the IC-program. \power changes the IC-Program each time the \power-protocol is initiated; the diversity renders the information learned by the attacker from the previous run obsolete. Moreover, the probability that a program will ever get repeated is $\sfrac{1}{10^{20}}$.

\textbf{Attacker Hiding}: In case the attacker attempts to hide, she has to predict when the \power-protcol is initiated. \power's random initiation mechanisms ensures that the attacker does not predict those instances. Our game theoretic analysis shows that using an exponential initiation strategy, the attacker's best strategy is to either always hide if the verifier is aggressive or attempt to hide without avail and get detected eventually.

\textbf{Forced Retraining}: In case the attacker forces \power to retrain by simulating a hardware fault that requires a CPU change, in order to lead \power to a compromised model. Then \power's process is to wipe the permanent storage, retrain using a clean OS, and then restore data. Since, we assume that the attacker does not modify the hardware state, then by removing permanent storage, she cannot impact the retraining process.

\section{Related Work}
\label{sec:related}

\subsubsection{Timing Attestation}
Seshadri et al. propose Pioneer~\cite{pioneer:2005} extended by  Kovah et. al.~\cite{pioneer:2012} a timing-based remote attestation system for legacy system (without TPM). The timing is computed using the network round trip time. The work assumes that the machine can be restricted to execution in one thread. The issue with the work is that the round trip time is affected by the network conditions which the authors do not explore, a heavily congested network will lead to a high variation on the RRT causing a high rate of false positives. Moreover, the restriction of execution in one thread can be evaded by a lower level attacker. In later work the authors discuss the issues of Time Of Check, Time Of Use attacks, we talk the problem in our work. Later work adapted timing attestation to embedded devices~\cite{min:atts:2014}.

Hern\'{a}ndez et. al.~\cite{phase:2015} implement a monitor integrity checking system by estimating the time it takes for a software to run. The timing information is sent from the machine to a remote server that uses a phase change detection algorithms to detect malicious changes. The issue of this work is that the timing information is sent by the untrusted machine and thus the information can be easily manipulated. Armknecht et. al.~\cite{framework:atts:2013} propose a generalized framework for remote attestation in embedded systems. The authors use timing as a method to limit the ability of an attacker to evade detection. The framework formalizes the goals of the attacker and defender. The authors provide a generic attestation scheme and prove sufficient conditions for provable secure attestation schemes.

\subsubsection{Power Malware Detection}
Several researchers use power usage to detect malware. In WattsUPDoc Clark~\cite{clark:whattsUpDoc} collect power usage data by medical embedded devices and extract features for anomaly detection. The authors exploit the regularity of the operation of an embedded device to detect irregularities. The authors however do not investigate mimicry attacks. Kim et. al.~\cite{Kim:2008} use battery consumption as a method to detect energy greedy malware. The power readings are sent from the untrusted device to a remote server to compare against a trusted baseline. The problem of this work is that the power readings can be manipulated by the attacker as the data is sent through the untrusted software. PowerProf~\cite{powerprof} is another in-device unsupervised malware detection that uses power profiles. The power information is similarly passed through the untrusted stack and is thus susceptible to attacker evasion through tampering. 


\subsubsection{Hardware Attestation}
Secure Boot~\cite{secboot} verifies the integrity of the system, with the root of trust a bootloader. Later on Trusted Platform Modules (TPMs) uses Platform Configuration Registers (PCRs) store the secure measurements (hash) of the system. Both methods are static in that the integrity is checked at boot time. Dynamic attestation on the other hand can perform attestation on the current state of the system. Such features are supported by CPU extensions (for example Intel TXT).
El Defrawy et al. propose SMART~\cite{smart:2012}, an efficient hardware-software primitive to establish a dynamic root of trust in an embedded processor, however the authors do not assume any hardware attack.

\subsubsection{VM based Integrity checker}
OSck~\cite{osck:2011} proposed by Hofmann et. al. is a KVM based kernel integrity checker that inspects kernel data structures and text to detect rootkits. The checker runs as a guest OS thread but is isolated by the hypervisor. Most VMM introspection intergrity checker assume a trusted hypervisor. Those techniques are vulnerable to hardware level attacks~\cite{kovah2015senter,PIkit, wojtczuk2009attacking}. In our work we do not have any trust assumption as the attestation device is external to the untrusted machine.

\subsubsection{Checksum Diversity}
Wang et. al.~\cite{Wang:2001} propose using diversity of probe software for security.
The authors obfuscate the control flow by flattening the probing software in order to make it harder for an attacker to reverse engineer the program for evasion.
While the flattened control flow is hard to statically analyze, the programs are susceptible to active learning thus allowing an attacker to adapt over time. Giffin et. al.~\cite{selfchange} propose self-modifying to detect modification of checksum code modification.
The experiments show an overhead of 1 microsecond to each checksum computation, the method is however costly for large programs adding second per check.
The authors in~\cite{WSN:atts:2009} use randomized address checking and memory noise to achieve unpredictability.


\section{Conclusion}
\label{sec:conclusion}
In this work, we presented \power an external integrity checker that uses power measurements as a trust base. The power signal provides an untainted, trusted, and very accurate method for observing the behavior of the untrusted computer.
\power initiates the interrogation protocol with a randomly generated integrity checking program. The diversity of the IC-program prevents the attacker from adapting. We show that the space of IC-programs is impossible to exhaust and that the generation is very efficient for low-power devices.
\power measures the current drawn by the processor during computation and compares it to a learned model to validate the output of the untrusted machine.
We model the interaction between \power and the attacker with as a time continuous game. The attacker disables her malicious activities at randomly chosen time instants in order to evade
\power's integrity checks.
Our simulations show that the attacker trades off stealthiness and her period of inactivity. An attacker that would want to remain stealthy needs to remain inactive for longer periods of time and an increase in her activity period leads to an increase in the probability of her being detected by \power. We also show that even of a stealthy attacker, \power still achieves an acceptable probability of detection given the long lifetime of stealthy APTs; by remaining stealthy, the attacker delays the inevitable and incurs extended periods of inactivity. Going forward we will formally prove the insights that the simulations have shown, and seek to find equilibrium and dominant strategies for \power. 




\bibliographystyle{IEEEtranS}
\bibliography{aaapa}

\end{document}